\documentclass{article}

\usepackage{arxiv}
\usepackage[table,xcdraw]{xcolor}
\usepackage{amsmath,amsthm,amssymb,amscd}
\usepackage[all,cmtip]{xy}
\usepackage{diagxy}
\usepackage[utf8]{inputenc} % allow utf-8 input
\usepackage[T1]{fontenc}    % use 8-bit T1 fonts
\usepackage{hyperref}       % hyperlinks
\usepackage{url}            % simple URL typesetting
\usepackage{booktabs}       % professional-quality tables
\usepackage{amsfonts}       % blackboard math symbols
\usepackage{nicefrac}       % compact symbols for 1/2, etc.
\usepackage{microtype}      % microtypography
\usepackage{lipsum}
\usepackage{mathrsfs}
\usepackage{fancyhdr}
\usepackage{color}
\usepackage{subfig}
\usepackage{graphicx}
\usepackage{textcomp}
\usepackage{ytableau}
\usepackage{tikz}
\usetikzlibrary{calc,intersections,through,backgrounds,patterns}
\usepackage{algorithm}
\usepackage{algpseudocode}
\newtheorem{exam.}{Example}
\newtheorem{def.}{Definition}
\newtheorem{definition}[subsubsection]{Definition}
\newtheorem{theorem}[subsubsection]{Theorem}
\newtheorem{corollary}[subsubsection]{Corollary}

\title{Quotient complex (QC)-based machine learning for 2D perovskite design}

\author{Chuan-Shen Hu \\ 
	Division of Mathematical Sciences\\
	School of Physical and Mathematical Sciences\\
        Nanyang Technological University \\
	Singapore 637371 \\
	\texttt{chuanshen.hu@ntu.edu.sg} \\
	\And
	Rishikanta Mayengbam \\
	Division of Physics \& Applied Physics\\
	School of Physical and Mathematical Sciences\\
        Nanyang Technological University \\
	Singapore 637371 \\
	\texttt{rishikanta.m@ntu.edu.sg} \\
	\AND
	Kelin Xia \\
	Division of Mathematical Sciences\\
	School of Physical and Mathematical Sciences\\
        Nanyang Technological University \\
	Singapore 637371 \\
	\texttt{xiakelin@ntu.edu.sg} \\
	\And
	Tze Chien Sum \\
	Division of Physics \& Applied Physics\\
	School of Physical and Mathematical Sciences\\
        Nanyang Technological University \\
	Singapore 637371 \\
	\texttt{tzechien@ntu.edu.sg} \\
}
\begin{document}
\maketitle

\begin{abstract}
With remarkable stability and exceptional optoelectronic properties, two-dimensional (2D) halide layered perovskites hold immense promise for revolutionizing photovoltaic technology. Presently, inadequate representations have substantially impeded the design and discovery of 2D perovskites. In this context, we introduce a novel computational topology framework termed the quotient complex (QC), which serves as the foundation for the material representation. Our QC-based features are seamlessly integrated with learning models for the advancement of 2D perovskite design. At the heart of this framework lies the quotient complex descriptors (QCDs), representing a quotient variation of simplicial complexes derived from materials unit cell and periodic boundary conditions. Differing from prior material representations, this approach encodes higher-order interactions and periodicity information simultaneously. Based on the well-established New Materials for Solar Energetics (NMSE) databank, our QC-based machine learning models exhibit superior performance against all existing counterparts. This underscores the paramount role of periodicity information in predicting material functionality, while also showcasing the remarkable efficiency of the QC-based model in characterizing materials structural attributes.
\end{abstract}

% keywords can be removed

\begin{figure}[t!]%[tbhp]
\centering
\includegraphics[width=1.0\linewidth]{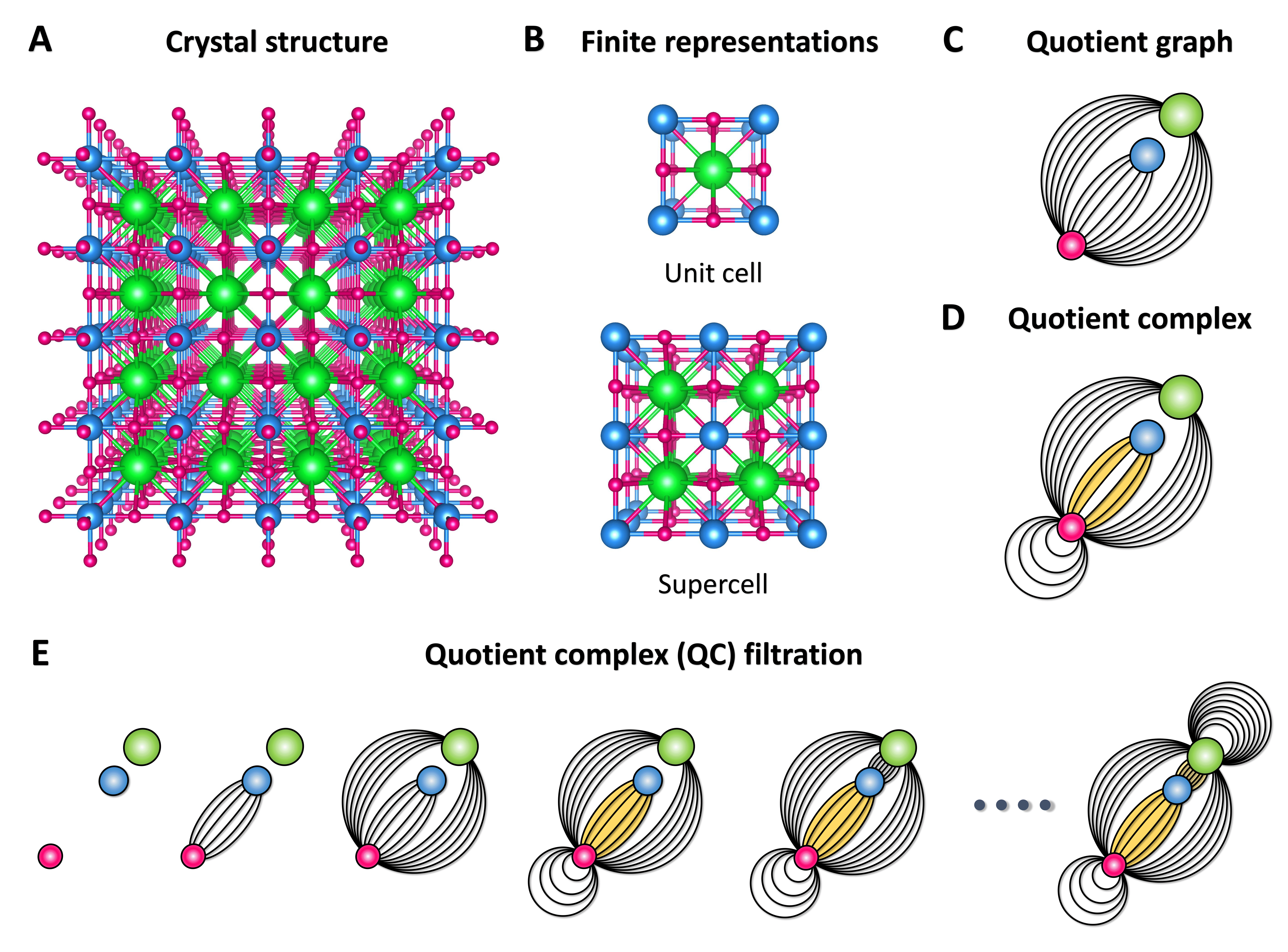}
\caption{Illustration of a crystal structure (Panel \textbf{A}) and its finite representations (Panel \textbf{B}), including the unit cell and supercell representations. The entire crystal structure is built by repeating the unit cell or supercell. In particular, the supercell depicted in \textbf{B} represents a $2 \times 2 \times 2$ extension of the unit cell. Based on the periodicity of duplicated atoms, a quotient graph representation is obtained (Panel \textbf{C}). Furthermore, by leveraging higher-dimensional objects (e.g., 2-dimensional surfaces) in the quotient graph, the quotient complex (Panel \textbf{D}) is established. Panel \textbf{E} illustrates a schematic image of a quotient complex filtration derived from the crystal structure. The crystal, unit cell, and supercell structures were visualized using the VESTA program~\cite{momma2011vesta}.}
\label{fig: Periodic structure}
\end{figure}

\section*{Introduction}
Materials innovation and design stand at the forefront of humanity's response to the grand challenge of achieving sustainable energy generation and combating climate change \cite{green2017fulfilling}. In less than a decade, perovskites have showcased remarkable power conversion efficiency (PCE), standing shoulder-to-shoulder with state-of-the-art technologies. The performance not only establishes them as a formidable contender but also hints at the possibility of a paradigm shift in photovoltaic (PV) technology \cite{grancini2019dimensional,mao2018two}. Especially, single-junction perovskite solar cells achieve an impressive 26\% efficiency, while perovskite-silicon tandem counterparts reach an even higher 33\%, outperforming single-junction silicon solar cells \cite{min2021perovskite, sharma2022stability}. In recent times, the advent of a diverse family of two-dimensional (2D) perovskites positions them as leading candidates for the next generation of perovskite materials. For instance, beyond the archetypal Ruddlesden-Popper (RP) phases, Dion–Jacobson (DJ), and alternating-cation-interlayer (ACI) structures have rapidly emerged~\cite{gong2022layered}. Nevertheless, the diverse range of organic cations available to the RP, DJ, and ACI phases gives rise to an incredibly intricate chemo-structural space, unlike their more straightforward 3D counterparts. This complexity not only makes traditional trial-and-error methods time-intensive but also remarkably costly and impractical for screening \cite{pilania2013accelerating}.

Over the past few years, data-driven materials informatics has demonstrated great potential for material design and discovery \cite{xie2018crystal,ramprasad2017machine,damewood2023representations}. Data-driven materials informatics involves the convergence of computational materials science with machine learning technologies and is regarded as the fourth paradigm of materials science \cite{hey2009fourth}. Notably, the Materials Genome Initiative (MGI) \cite{national2011materials}, initiated in 2011 within the United States, represents a pivotal milestone in data-driven materials informatics. MGI's ambitious goal was to foster extensive collaboration between materials scientists and computer scientists, harnessing the power of machine learning models to predict, screen, and optimize materials on an unprecedented scale and pace. This initiative spurred the establishment of pivotal databanks such as the Materials Project \cite{jain2013commentary}, JARVIS \cite{choudhary2020joint}, NOMAD, Aflowlib, and OQMD. Collaborating with other well-established materials repositories, they lay a robust foundation for data-centric artificial intelligence (AI) models. These models play a pivotal role in predicting materials' functionality, propelling forward the realms of new material design, and accelerating ground-breaking material discoveries. At present, AI models have been widely used in the analysis of various aspects of perovskite materials, encompassing single perovskites \cite{balachandran2018predictions, li2018stability, park2019exploring, schmidt2017predicting}, double perovskites \cite{ l2019machine, li2019thermodynamic}, lead (Pb)-free perovskites \cite{im2019identifying, jacobs2019materials, wu2019global}, and perovskite solar cell devices \cite{li2019predictions, odabacsi2020assessment}.

Crystal periodic information is crucial in shaping the physical properties of a material. In material science, periodic information is considered through two primary approaches. The first approach focuses on integrating periodicity into structural, chemical, and physical features and descriptors~\cite{rupp2012fast,huo2022unified,bartok2013representing}. For instance, the smooth overlap of atomic positions (SOAP) features, specifically designed for crystal structures, take into account both periodic and electronic interactions between translated unit cells~\cite{bartok2013representing}, and found great success in 2D perovskite bandgap prediction. In periodic geometry, periodic structural information is incorporated into crystal features, such as density fingerprints and density functions. Periodic geometry models provide a multiscale, continuous, and periodically invariant descriptor for the crystalline structure ~\cite{edelsbrunner2021density, anosova2023density, widdowson2023recognizing, anosova2022algorithms, kurlin2023polynomial}. The second approach is to develop a quotient topological representation that incorporates periodic information, which is known as topological crystallography and crystal-based discrete geometric analysis ~\cite{sunada2012topological,sunada2012lecture,sunada2016discrete}. Among them is the \textit{quotient graph}, which plays a pivotal role as a representation of the periodic crystalline structure. In periodic graph representation, graph symmetries, loop structures, and periodic tiling patterns all characterize material's periodic information~\cite{o2022symmetry,o2022tangled,delgado2023three,gao2020determining}.

Recently, topological data analysis (TDA) has found great success in the analysis of molecular data from biology, chemistry, and materials~\cite{cang:2017integration, cang:2017topologynet, cang:2018representability, nguyen2020review,wu:2018quantitative,wu2018topp,anand2022topological}. TDA-based molecular features have three unique properties that differ dramatically from previous material descriptors. First, TDA makes use of \textit{simplicial complexes}, which is a generalization of graphs into their higher-dimensional counterparts. Other than pair-wise interactions, simplicial complexes can also characterize high-order or many-body interactions. Second, a filtration process is used in TDA and naturally incorporates it into a multiscale representation, which can be suitable for the modeling of various ionic interactions, van der Waals interactions, etc. Third, TDA employs the topological invariant of homology as the descriptors. These invariants are more intrinsic and fundamental, thus providing a better transferability for learning models.  More recently, a periodic simplicial complex has been proposed to characterize the infinite material structure spanned by the unit cell under periodic conditions.
Essentially, driven by the need for materials characterization, the quantization of the homological structures of infinite periodic simplicial complexes has emerged as a crucial research focus in computational topology~\cite{Onus2022, onus2023computing, carletti2023global}. Specifically, certain methods for quantifying the homological groups of periodic simplicial/cell complexes, particularly focusing on loop structures within infinite material arrangements, have been proposed~\cite{Onus2022, onus2023computing}.

Here we propose \textit{quotient complex} (QC)-based material representation and QC-based machine learning models for 2D perovskite design. Our QC-based learning framework integrates several critical elements, including high-order representation with simplicial complexes, quotient complexes derived by encoding the periodicity to the underlying simplicial complex, and a multiscale feature representation based on PH. Specifically, we utilize the persistent homology (PH) of QC filtration as a representation of the 2D perovskite structure. Through this representation, we define a novel descriptor termed QC-based descriptor (QCD) based on the induced PH. In particular, the QC-based approach would be highly suitable for the elaborate family of 2D layered perovskites. Subsequently, we integrate these descriptors into a gradient boosting tree (GBT) model to predict materials' bandgaps. Our QC-GBT (quotient complex-based GBT) model demonstrates remarkable performance, comparable to, and in some cases even surpassing, existing state-of-the-art models, as validated extensively. These evaluations are conducted on the 2D perovskite dataset obtained from the new materials for solar energetics (NMSE) databank~\cite{marchenko2020database}, affirming the efficacy of our proposed approach.

\section*{Results}
\label{sec:Results}

\begin{figure*}[t!]%[tbhp]
\centering
\includegraphics[width=1.0\linewidth]{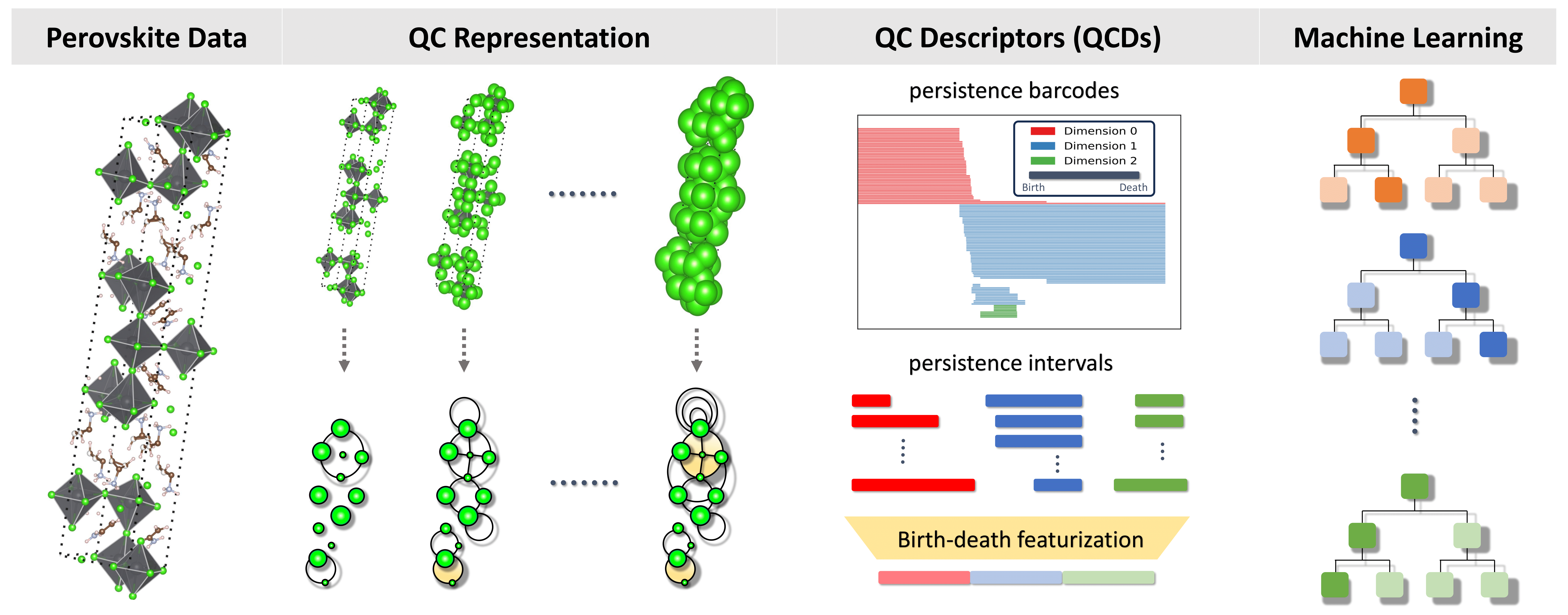}
\caption{Flowchart illustrating the quotient complex (QC)-based ML model used in this study to predict material bandgaps for 2D perovskite structures. In this work, we atomwisely compute the QC of a given material structure. The second panel of the flowchart depicts spheres centered at the iodine atoms of the material at varying radii, accompanied by the schematic visualization of the corresponding QC filtration. The third panel illustrates the persistence barcodes (PBs) of the QC filtration. Each PB collects intervals (or bars) with beginning and ending values based on the multiscale process, referred to as the birth and death of the interval. By examining the birth and death of intervals within the PB, the feature vector is generated, which serves as an input vector for the ML models. Finally, a gradient boosting tree (GBT) model is applied to the QC-based descriptors to predict material bandgaps. Perovskite structures were visualized using the VESTA program~\cite{momma2011vesta}.}
\label{fig:flowchart}
\end{figure*}

\subsection*{Quotient complex (QC)-based material representation}

We propose a multiscale representation of periodic crystalline structures based on quotient complexes (QCs). Illustrated in Fig. \ref{fig: Periodic structure}, a crystalline arrangement replicates unit cells of atoms (Fig. \ref{fig: Periodic structure}\textbf{A} and \textbf{B}), constituting the entire material. Using the coordinates of the periodic atomic systems, graphs or simplicial complexes are constructed to depict the geometric structure of the material. Furthermore, through the consolidation of periodic points within the system, we obtain the quotient graph and quotient complex (Fig. \ref{fig: Periodic structure}\textbf{C} and \textbf{D}). Ultimately, employing a multiscale approach to create a filtration of simplicial complexes (such as the Vietoris--Rips complexes), the induced filtration of QCs is chosen as the representation of the material (Fig. \ref{fig: Periodic structure}\textbf{E}).

Mathematically, the key components of our QC-based representation are quotient operation and filtration process. Consider the periodic vertex set $V$ as a subset of $\mathbb{R}^3$, representing the atomic system. This system is constructed by iteratively duplicating a local atomic arrangement referred to as a \textit{motif} within a \textit{unit cell}. Formally, a unit cell is defined as the parallelepiped generated by a basis $\{ v_1, v_2, v_3 \} \subseteq \mathbb{R}^3$, expressed as:
\begin{equation*}
U = U(v_1, v_2, v_3) = \left \{ \sum_{i=1}^3 a_iv_i \ \bigg| \ a_i \in [0,1) \right \},
\end{equation*}
and a motif $M$ is considered as a finite subset of $U$. Based on the basis $\{ v_1, v_2, v_3 \}$, the associated \textit{lattice group} is defined by
\begin{equation*}
\Lambda = \Lambda(v_1, v_2, v_3) = \left \{ \sum_{i=1}^3 a_iv_i \ \bigg| \ a_i \in \mathbb{Z} \right \}.
\end{equation*}
Then the periodic set $V$ is defined as the Minkowski sum
\begin{equation*}
V = M + \Theta = \bigcup_{w \in \Theta} M + w,
\end{equation*}
where $\Theta$ is a finite subset of $\Lambda$. In other words, $V$ is a union of finitely many translated copies of $M$ and plays a pivotal role in real applications, such as the supercell representation of the material~\cite{krishnapriyan2020topological, anand2022topological}. With the set $V$, two vertices $v$ and $w$ are regarded as periodic points, denoted as $v \sim_V w$, if either of them can be translated by vectors in the lattice $\Lambda$. That is,
\begin{equation}
\label{Eq. The periodic equivalence relation}
v \sim_V w \ \text{ if, and only if} \ v-w \in \Lambda.
\end{equation}
In particular, this translation relation forms an equivalence relation on the vertex set $V$. A QC with the vertex set $V$ can be constructed as follows. Let $K$ be a simplicial complex constructed from $V$, for instance, the Vietoris--Rips complex $K = K_\epsilon(V)$ generated using a fixed radius $\epsilon$~\cite{vietoris1927hoheren,chazal2013interleaved}:
\begin{equation*}
K_\epsilon(V) = \bigcup \{ {\rm conv}(X) \ | \ {\rm diam}({\rm conv}(X)) \leq \epsilon\text{, } X \subseteq_{\text{finite}} V  \},
\end{equation*}
where ${\rm conv}(X)$ denotes the convex hull of a finite set $X \subseteq V$ in $\mathbb{R}^3$ and ${\rm diam}({\rm conv}(X)) = \max\{ \Vert x-y \Vert  \ | \  x, y \in X \}$ is the diameter of ${\rm conv}(X)$. Regarding $\sim_V$ as an equivalence relation on the whole $K$, the QC is defined as the quotient topological space $\overline{K} := K/\sim_V$.

Further, a multiscale representation of the crystalline structure is achieved through a filtration process. Specifically, let $K_{\epsilon_1} \subseteq K_{\epsilon_2} \subseteq \cdots \subseteq K_{\epsilon_n}$ be a Vietoris--Rips filtration over the same $V$, with radii $\epsilon_1 < \epsilon_2 < \cdots < \epsilon_n$. Then, each $K_i$ induces a QC $\overline{K_{\epsilon_i}}$, and we obtain a sequence of QCs $\overline{K_{\epsilon_1}}, \overline{K_{\epsilon_2}}, \ldots, \overline{K_{\epsilon_n}}$. Indeed, as the equivalence relations $\sim_V$ remain the same across all $K_{\epsilon_i}$, there exists a canonical inclusion map $\overline{K_{\epsilon_i}} \hookrightarrow \overline{K_{\epsilon_{i+1}}}$ for each $i = 1, 2, ..., n-1$. Consequently, we obtain the following filtration of QCs:
\begin{equation*}
\emptyset \subseteq \overline{K_{\epsilon_1}} \subseteq \overline{K_{\epsilon_2}} \subseteq \cdots \subseteq \overline{K_{\epsilon_n}}.
\end{equation*}

\begin{table*}[t!]
\centering
\caption{Comparison of the performances of methods in predicting 2D perovskites' DFT-based bandgaps based on the NMSE database~\cite{marchenko2020database}. The compared methods include SOAP-based models (KRR~\cite{mayr2021global} and  MLM1~\cite{marchenko2020database}) and GNN models (GCN~\cite{na2023substructure}, ECCN~\cite{na2023substructure}, CGCNN~\cite{na2023substructure}). The numerical count of materials used in various experiments is presented in the second column. Performance assessments of GCN, ECCN, CGCNN, TFGNN, SIGNNA, and SIGNNA$_c$ were replicated following Na's work~\cite{na2023substructure}. The evaluation metrics encompass the coefficient of determination (COD), Pearson correlation coefficient (PCC), mean absolute error (MAE), and root-mean-square error (RMSE). These metrics represent average scores derived from 5-fold cross-validation in $5$ times. Bold and underlining notations are used to signify the most and second exceptional results, respectively.}
\begin{tabular}{|lccccc|}
 \hline
 Method & Number of Compounds & COD & PCC & MAE (eV) & RMSE (eV) \\ % [0.5ex]
 \hline\hline
 SOAP-KRR~\cite{mayr2021global} & 445 & 0.6700 & - & 0.1250 & 0.2530 \\
 SOAP-MLM1~\cite{marchenko2020database}  & 515 & 0.9005  & - &  0.1030 &  0.1360 \\
 \hline
 GCN~\cite{na2023substructure} & 624 & 0.6030 & - & 0.1910 & - \\
 ECCN~\cite{na2023substructure} & 624 & 0.3170 & - & 0.2240 & - \\
 CGCNN~\cite{na2023substructure} & 624 & 0.6450 & - & 0.1750 & -  \\
 TFGNN~\cite{na2023substructure} & 624 & 0.6650 & - & 0.1620 & - \\
 SIGNNA~\cite{na2023substructure} & 624 & 0.9080 & - & 0.0920 & - \\
 SIGNNA$_c$~\cite{na2023substructure} & 624 & \textbf{0.9270} & - & 0.0830 & - \\
 \hline
 QC-GBT  & 445 & 0.9111 & 0.9558 & 0.0751 & 0.1102 \\
 QC-GBT  & 515 & 0.9092 & 0.9549 & 0.0754 & 0.1139 \\
 QC-GBT  & 624 & 0.9207 & \textbf{0.9613} & \underline{0.0722} & \underline{0.1086} \\
 QC-GBT  & 716 & \underline{0.9213} & \underline{0.9609} & \textbf{0.0700} & \textbf{0.1077} \\
 \hline
\end{tabular}
\label{Table: Performances}
\end{table*}

\subsection*{Quotient complex (QC)-based material descriptor}

Based on the filtration process, we obtain the following QC-based PH in degree $q = 0, 1, 2$:
\begin{equation*}
{\rm PH}_q(\overline{K_{\epsilon_\bullet}}): 0 \rightarrow H_q(\overline{K_{\epsilon_1}}) \rightarrow H_q(\overline{K_{\epsilon_2}}) \rightarrow \cdots \rightarrow H_q(\overline{K_{\epsilon_n}}).
\end{equation*}
From the persistent homology ${\rm PH}_q(\overline{K_{\epsilon_\bullet}})$ information, we can generate the associated \textit{quotient complex descriptors} (QCDs) for the crystal structure. Specifically, we calculate the persistence barcode (PB) of ${\rm PH}_q(\overline{K_{\epsilon_\bullet}})$, which is a multiset of persistence intervals $(b,d)$ with real numbers $b < d$, recording the birth and death of $q$-dimensional homological structures in the filtration~\cite{CZCG05,Ghrist:2008}.
In practice, statistical measures of birth/death information, such as birth/death moments, lifespans, and midpoints, along with Betti curves, are utilized to summarize the $q$-dimensional barcodes with $q = 0, 1, 2$. These serve as descriptors for the crystalline material and form the input features for the machine learning model. Additional details regarding the QC filtration, PH of the QC filtration, and the QCDs are provided in the Materials and Methods section and the Appendix \ref{Appendix: Quotient Complex Descriptors (QCDs)}.

The crystal structure's QC-based material representation and descriptors seamlessly encode the interactions of atoms and the associated periodicity of the crystalline structure, playing a crucial role in shaping the material's structural properties. In particular, going beyond the interactions between paired atoms within the unit cell, our proposed QC structure takes into account the intrinsic and complex ``intra-interactions'' of atoms in groups containing more than three atoms, utilizing the based simplicial complex structure. Moreover, the QC structure adeptly incorporates the periodic information within the material into its homological structure. For instance, the self-loops (e.g., Fig. \ref{fig: Periodic structure}\textbf{D}) within the QC representation notably depict the ``inter-interactions'' of periodic atoms, capturing the geometric relationships between distinct periodic cells. Moreover, the more intrinsic inter-interactions among cells may be reflected through more intricate loops in QC. From this perspective, elements in the homology groups $H_q(\overline{K})$ serve not only to discern the global topology of the simplicial complex $K$ but also to depict the local geometry and periodicity inherent in the crystalline structure. In particular, we conduct a theoretical analysis of the persistence intervals within the PHs $H_q(K_{\epsilon_\bullet})$ and $H_q(\overline{K_{\epsilon_\bullet}})$ and classify them based on their lifespans. Especially, for the $1$-th PH we prove that
\begin{equation*}
\begin{split}
{\rm PB}_1(\overline{K_{\epsilon_\bullet}}) &= {\rm PB}_1^{\rm finite}(\overline{K_{\epsilon_\bullet}}) \cup {\rm PB}_1^{\infty}(\overline{K_{\epsilon_\bullet}}) \\
&= {\rm PB}_1(K_{\epsilon_\bullet}) \cup {\rm PB}_1^{\infty}(\overline{K_{\epsilon_\bullet}}),
\end{split}
\end{equation*}
where ${\rm PB}_1^{\rm finite}(\overline{K_{\epsilon_\bullet}})$ and ${\rm PB}_1^{\infty}(\overline{K_{\epsilon_\bullet}})$ represent the collections of finite and infinite persistence intervals within the QC filtration. These intervals correspond to the global topology structure and the periodicity information of the material, respectively.  Theoretical discussions and proofs are presented in Appendix \ref{Appendix: Mathematical Foundation of QC Model}.

\subsection*{QC-based machine learning for 2D perovskites}
The flowchart of our QC-based learning model for 2D perovskite analysis can be found in Fig. \ref{fig:flowchart}. The methodology comprises three primary phases. Firstly, a diverse array of atom types and sites within the unit cell, encompassing organic atoms (e.g., $\rm O$), inorganic atoms (e.g., $\rm Sn$), halide atoms (e.g., $\rm Cl$), organic site ($\rm A$), inorganic site ($\rm B$), halide site ($\rm X$), and their combinations (e.g., ${\rm A}_{\rm C}{\rm B}{\rm X}$, which collects ${\rm C}$ atoms to represent the organic site) are systematically extracted from the dataset. Subsequently, a cohesive augmentation of atom radii is achieved through a filtration process to establish a multiscale cell complex filtration of QCs. This filtration yields a collection of 1D feature vectors generated based on its PH. Finally, leveraging the remarkable accuracy and robustness of the Gradient Boosted Tree (GBT) model, we utilize these feature vectors to predict material bandgaps~\cite{friedman2001greedy,salami2022estimating, piryonesi2020data,chun2021stroke}.

We employ standard statistical measures to assess the performance of our QC-GBT model, utilizing four key indices for error evaluation~\cite{marchenko2020database, mayr2021global, na2023substructure}: the coefficient of determination (COD), Pearson correlation coefficient (PCC), mean absolute error (MAE), and root-mean-square error (RMSE). The same 5-fold cross-validation is employed for all experiments. Table \ref{Table: Performances} lists the results. It can be seen that our QC-GBT model has a COD of 0.9207 and an MAE of 0.0722 eV, highlighting its superior performance.

\subsection*{Extensive comparison with state-of-the-art models}
We conduct an extensive comparative analysis of the proposed QC-GBT against state-of-the-art models for predicting the bandgap of 2D perovskites. This evaluation encompasses models based on the smooth overlap of atomic positions (SOAP) kernel, namely SOAP-based KRR~\cite{mayr2021global} and MLM1~\cite{marchenko2020database}), as well as Graph Neural Network (GNN) models including GCN \cite{welling2016semi}, ECCN \cite{park2020developing}, CGCNN \cite{xie2018crystal}, TFGNN \cite{shi2020masked}, and SIGNNA \cite{na2023substructure}. %A comprehensive breakdown of their performances is outlined in Table \ref{Table: Performances}, with each model utilizing the same 5-fold cross-validation procedure.

The proposed model can outperform molecular feature-based ML models based on the mean absolute error (MAE) values (eV). Notably, the SOAP descriptor-driven GBT model, namely SOAP-MLM1 \cite{marchenko2020database}, has recently been employed for predicting 2D perovskite bandgaps. This model exhibits a coefficient of determination (COD) of approximately 0.90 and an MAE of around 0.103 eV. Additionally, the SOAP descriptor has been integrated with an autoencoder network and the kernel ridge regression (KRR) technique to facilitate feature reduction and bandgap prediction, primarily in the context of 3D perovskite datasets \cite{kim2017hybrid,mayr2021global}. However, it is worth noting that, when evaluated against the NMSE database, the SOAP-KRR model's performance falls short compared to MLM1.

Our QC-GBT models consistently outperform in comparison to the GNN models in Table \ref{Table: Performances}, as evidenced by lower MAE values. Recent advancements in GNN-based models listed in Table \ref{Table: Performances} have showcased remarkable capabilities in handling 3D perovskite datasets~\cite{mayr2021global, anand2022topological, schutt2017schnet, chen2019graph}. In a comprehensive assessment against GNN-based methods, namely CGCNN~\cite{mayr2021global}, GCN~\cite{welling2016semi}, ECCN~\cite{park2020developing}, TFGNN~\cite{shi2020masked},  SIGNNA~\cite{na2023substructure}, and SIGNNA$_c$~\cite{na2023substructure} carried out using the NMSE database, our models consistently exhibit a better performance. In particular, SIGNNA$_c$ is an enhanced version of SIGNNA that incorporates additional chemical information and machine-learning feature vectors within the SIGNNA architecture. In comparison to SIGNNA, SIGNNA$_c$ achieves improved results with a COD of 0.9270 and MAE of 0.0830 (eV). While SIGNNA$_c$ performs slightly better than QC-GBT in terms of COD, QC-GBT still outperforms SIGNNA$_c$ with a better MAE of 0.0722 (eV), even when applied to smaller datasets (e.g., 0.0751 eV). Overall, fingerprint-based models consistently outperform these GNN models, possibly due to limitations in training data volume and the increased complexity of unit cell structures.

\begin{table}
\centering
\caption{Predictions for bandgaps in five new materials within NMSE are made using the QC-GBT model, with computed DFT bandgaps~\cite{marchenko2020database} serving as new prediction labels. The indices and chemical formulas of the tested materials in NMSE are presented in the first and second columns. In the third column, DFT-computed bandgaps (eV) for these materials, as determined in our study, are displayed. The fourth column showcases the prediction outcomes of the QC-GBT model. These predictions represent average scores obtained over 5 iterations of training and testing. The training dataset for this experiment encompasses 716 materials from NMSE, each associated with a known DFT bandgap as documented in the NMSE database.}
\begin{tabular}{|cccc|}
 \hline
 No. & Formula & DFT bg (eV) & QC-GBT (eV) \\ % [0.5ex]
 \hline\hline
 110 &
 $[({\rm C}_6{\rm H}_{11}){\rm P}{\rm H}_3]_2{\rm Sn}{\rm I}_4$ &  1.3540 & 2.1087
 \\
 309 &  $[{\rm C}_2{\rm H}_5{\rm N}{\rm H}_3]_2{\rm Fe}{\rm Cl}_4$ & 2.7002 & 3.2804 \\
 315 &
 ${\rm Cs}_3{\rm Sb}_2{\rm I}_9$ & 1.5450 & 2.1718
 \\
 363 &  $[{\rm C}{\rm H}_3{\rm N}{\rm H}_3]_2{\rm Fe}{\rm Cl}_4$ & 3.1661 & 3.2603
 \\
 452 &  ${\rm Tl}_3{\rm Bi}_3{\rm I}_9$ & 2.3780 & 2.2812
 \\
 \hline
\end{tabular}
\label{Table: Testing on real data with details}
\end{table}

\subsection*{Bandgap prediction on new peroveskites}
To assess the predictive capabilities of the QC-GBT model, we applied it to predict the bandgaps of new 2D perovskites within the NMSE database. More specifically, we use DFT to evaluate bandgaps for perovskites in the NMSE database, that lacked DFT bandgap values. Five perovskites with indices 110, 309, 315, 363, and 452, are considered. Table \ref{Table: Testing on real data with details} illustrates the calculated and predicted bandgap values for five distinct compounds with indices 110, 309, 315, 363, and 452 that are annotated in NMSE. Treating the calculated DFT bandgaps as test labels, the MAE, COD, and PCC for QC-GBT are 0.4297 (eV), 0.43, and 0.879, respectively. While these MAE values are notably higher than those presented in Tables \ref{Table: Performances}, this disparity can be explained by the significantly smaller size of the test set in this experiment. Notably, QC-GBT's predictions for materials 452 (3.2603 eV) and 363 (2.2812 eV) closely match the DFT bandgap labels (3.1661 and 2.378 eV), with absolute errors of 0.0942 and 0.0967 eV, respectively. In contrast, the predicted value for material 309 is 3.2804 eV, resulting in an absolute error of 0.5802 eV compared to the DFT label (2.7002 eV).

\section*{Methods}
\begin{figure*}[t!]%[tbhp]
\centering
\includegraphics[width=1.0\linewidth]{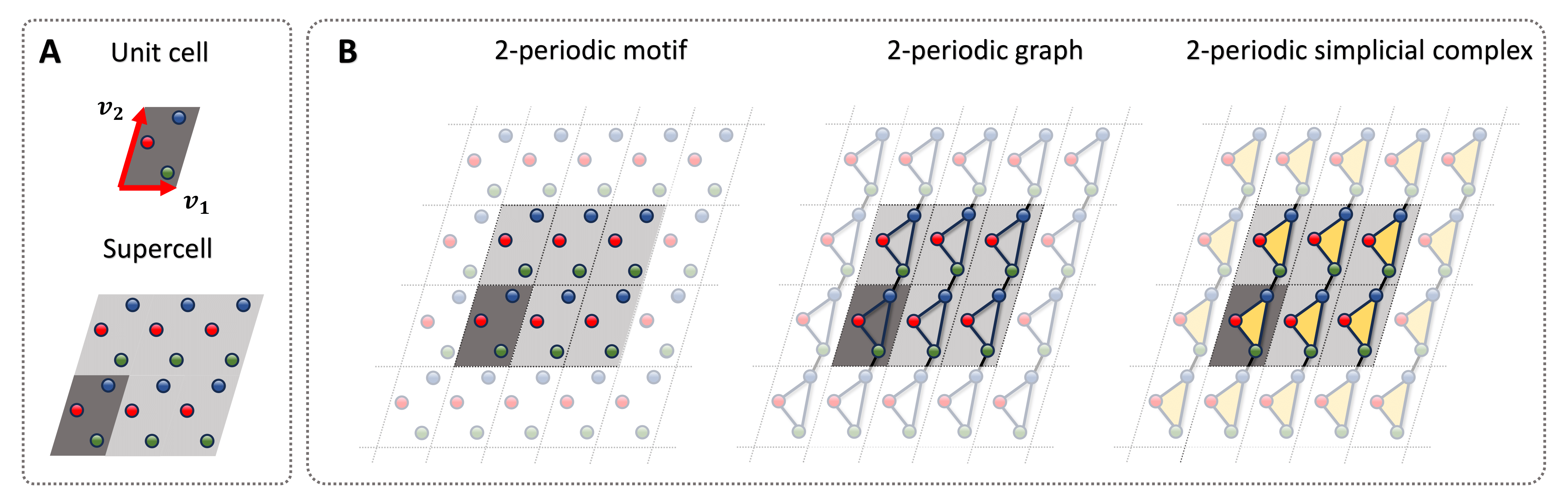}
\caption{A 2D illustration includes a unit cell, supercell, periodic motif, periodic simplicial complex, and a periodic graph. In Panel \textbf{A}, the red arrows ($v_1$ and $v_2$) accompanying the unit cell specifically represent a basis within 2-dimensional Euclidean space $\mathbb{R}^2$. In comparison to the shaded region of the supercell, the darker region denotes the unit cell containing $3$ points, forming the parallelepiped spanned by the basis vectors. In this example, a $3 \times 2$ supercell is illustrated, comprising a union of $6$ translated cells, totaling $3 \times 3 \times 2 = 18$ points. Panel \textbf{B} displays examples of $2$-dimensional periodic objects derived from unit cell information: a $2$-periodic motif, a $2$-periodic graph, and a $2$-periodic simplicial complex. The shaded regions within these objects delineate finite $d$-periodic objects, which constitute the primary focus of this study.}
\label{fig:periodic objects}
\end{figure*}

\subsection*{Mathematical foundation for QC-based material representation}
Mathematics serves as a fundamental pillar of materials science, providing a robust framework for representing crystalline data at the molecular level. Material structures are usually represented as unit cells, motifs, and space groups, which have led to the widespread adoption of the crystallographic information file (CIF) as the standard format~\cite{hall1991crystallographic}. Moreover, through the mathematical expressions of unit cells, motifs, and supercells, the corresponding simplicial and the proposed quotient complexes can be rigorously formulated. This section systematically introduces these mathematical formulations. More detailed analyses, discussions, and theoretical proofs are provided in Appendix \ref{Appendix: Mathematical Foundation of QC Model}.

\subsubsection*{Basic mathematical setting for material structures}
Rigorous mathematical formulations for unit cells, motifs, and supercells in $\mathbb{R}^3$ are of key importance for material structure representation. Specifically, a $d$\textit{-dimensional lattice group} is a group of the form
\begin{equation*}
\Lambda = \Lambda(\mathcal{B}) =  \mathbb{Z}v_1 \oplus \mathbb{Z}v_2 \oplus \cdots \oplus \mathbb{Z}v_d \simeq \mathbb{Z}^d,
\end{equation*}
where $\mathcal{B} = \{ v_1, v_2, ..., v_d \} \subseteq \mathbb{R}^d$ represents an $\mathbb{R}$-linear basis of $\mathbb{R}^d$. In crystalline analysis, the number $d$ is commonly considered to be $1, 2,$ or $3$. Subsequently, based on the basis $\mathcal{B}$, the \textit{unit cell}
\begin{equation*}
U = U(\mathcal{B}) := \left \{ \sum_{i = 1}^d c_iv_i \ \bigg| \ c_i \in [0,1) \right \}
\end{equation*}
is defined as the $d$-dimensional half-open parallelepiped generated by the vectors $v_1, v_2, ..., v_d$. Furthermore, a finite subset $M$ within the unit cell is called a \textit{motif}, which forms a unit atomic system that generates the entire material through translations. Mathematically, a $d$\textit{-periodic motif} (or a $d$\textit{-periodic set}) is expressed by the Minkowski sum
\begin{equation*}
S = M + \Theta,
\end{equation*}
where $\Theta$ is any subset of the lattice group $\Lambda$. In particular, when $\Theta = \Lambda$, then $S$ is an infinite set. From a computational perspective, we often consider a finite $\Theta$ and build the quotient complex representation on this finite atomic system.

\subsubsection*{Periodic simplicial complex and quotient complex for material representation}

\paragraph{Periodic simplicial complex model}
Going beyond the set of discrete points in $\mathbb{R}^d$, simplicial complexes in $\mathbb{R}^d$ can also exhibit periodic structures. Formally, a simplicial complex $K$ embedded in $\mathbb{R}^d$ is called a $d$\textit{-dimensional periodic simplicial complex} (or just a $d$\textit{-periodic simplicial complex}) if
\begin{equation*}
K = L + \Theta,
\end{equation*}
where $L$ is a finite simplicial complex and $\Theta$ is a subset of $\Lambda$~\cite{Onus2022,onus2023computing}. Based on the definition, the finiteness of $K$ is determined by the choice of $\Theta$, typically regarded as a finite set when constructing and computing simplicial complexes on a finite atomic system. On the other hand, when $\Theta = \Lambda$, there exists a group action of $\Lambda$ that transits effectively on $K$. In other words, $v+\sigma$ forms a simplex of $K$ whenever $v \in \Lambda$ and $\sigma$ is a simplex of $K$. This aligns with the settings when contemplating infinite $d$-periodic simplicial complexes~\cite{Onus2022, onus2023computing}.

To concisely summarize periodic objects, Fig. \ref{fig:periodic objects} visualizes the concepts of the unit cell, supercell, periodic motif, periodic graph, and periodic simplicial complex. In Panel \textbf{A}, the darker gray region specifically represents a 2-dimensional unit cell consisting of three points. The red arrows accompanying the unit cell form a basis $\mathcal{B} = \{ v_1, v_2 \}$ of $\mathbb{R}^2$ and span the unit cell $U(\mathcal{B})$ as a 2-dimensional parallelepiped. Within the unit cell, the three points form a 2-dimensional motif $M$, and the entire $2$-periodic motif is represented as the first image in Panel \textbf{B}. Mathematically, the entire motif can be represented by $S = M + \Lambda$, where
\begin{equation*}
\Lambda = \{ av_1 + bv_2 \ | \ a, b \in \mathbb{Z} \}
\end{equation*}
is the 2-dimensional lattice group generated by the basis $\mathcal{B}$. Based on the lattice group $\Lambda$, any two points are annotated with the same color (red, blue, and green) if they can be translated to each other by a vector in $\Lambda$. Furthermore, as a special case of finite periodic motifs, the $3 \times 2$ supercell illustrated in Panel \textbf{A} is a union of 6 translated cells, which can be formulated by
\begin{equation*}
V = \bigcup_{w \in \Theta} M + w,
\end{equation*}
where $\Theta = \{ av_1 + bv_2 \ | \ a \in \{ 0, 1, 2, 3 \}, b = \{ 0, 1, 2 \} \}$. From a computational standpoint, the set $V$ gathers finitely many points as a point cloud, offering a platform for the application of data analysis tools (e.g., PH).

In Panel \textbf{B}, besides the periodic motif, two examples of periodic objects are highlighted: a $2$-periodic graph and a $2$-periodic simplicial complex. Within the periodic simplicial complex, the 2-simplices are represented as yellow triangles. A periodic graph, as a special case of periodic simplicial complexes, encompasses vertices and edges as components of the 0 and 1-simplices. As depicted in these two examples, the translation of each $q$-simplex ($q = 0, 1, 2$) also belongs as a part of the entire graph and simplicial complex, respectively. This translation action on simplices defines the periodicity of these two geometric objects. Notably, the same periodic relations of simplices can be defined for finite simplicial complexes, forming the central focus of this study. As depicted in Panel \textbf{B}, the graph and simplicial complex bounded by the $3 \times 2$ supercell (the shaded region) exemplify a $2$-period graph and a $2$-period simplicial complex within a bounded region.

\paragraph{Quotient complex model}
Based on the translation relations on simplices, we introduce a framework known as the quotient complex (QC) representation for any simplicial complex with periodic information. Specifically, for any simplicial complex $K \subseteq \mathbb{R}^d$ that contains a $d$-periodic subcomplex $A$, we define an equivalence relation $\sim_A$ as follows:
\begin{equation}
\label{Eq. Equivalence relation on A}
x \sim_A y \text{ if, and only if } x - y \in \Lambda(\mathcal{B}).
\end{equation}
Viewing $\sim_A$ as an equivalence relation on the entire $K$, we may define the quotient complex $\overline{K}$ as the quotient topological space $\overline{K} = K/\sim_A$. Based on the equivalence relation, points $x$ and $y$ in $A$ with $x \sim_A y$ are glued as the same point $[x] = [y]$ in the space $\overline{K}$. Alternatively, points belonging to $K \setminus A$ remain fixed under the equivalence relation $\sim_A$. It is noteworthy that although the equivalence relation is induced by the infinite lattice group $\Lambda(\mathcal{B})$, the equivalence relation defined in \eqref{Eq. Equivalence relation on A} is also valid on a finite $d$-periodic simplicial complex $A$. In applications, it provides a more flexible and feasible approach for computing data from a finitely sampled point cloud.

Fig. \ref{fig: Quotient Complexes} illustrates three examples of quotient complexes obtained from different equivalence relations on a finite 2-periodic simplicial complex. In this figure, the base simplicial complex $K$ consists of 6 vertices, 7 edges, and 2 triangles and is contained in a $1 \times 2$ supercell embedded in $\mathbb{R}^2$. Subcomplexes $G$ and $V$ of $K$ are defined as the union of all edges and vertices of $K$, and both of them admit equivalence relations of translations on their simplices. Let $\sim_G$ and $\sim_V$ be the equivalence relations of translations on $G$ and $V$, respectively. Then, the corresponding quotient complexes $K/\sim_K$, $K/\sim_G$, and $K/\sim_V$ are depicted in Fig. \ref{fig: Quotient Complexes}, showing different quotient topologies. In detail, $q(K) = K/\sim_K$ has Betti numbers (cf., \cite{munkres2018elements}) $(\beta_0, \beta_1, \beta_2) = (1,1,0)$, $K/\sim_G$ has $(1,1,1)$, and $K/\sim_V$ has $(1,3,0)$, showing the different loops and void structures.
In this work, our main focus will be on exploring quotient complexes of the form $K/\sim_V$ from a computational perspective.
 
\subsection*{QC-based learning model for 2D perovskite design}
\subsubsection*{Database}
The production and characterization of 2D perovskites pose intricate challenges, leading to a lack of comprehensive databases concerning these materials. However, a notable breakthrough occurred in 2021 when the Laboratory of New Materials for Solar Energetics (NMSE) initiated a pioneering crowd-sourcing endeavor. This initiative aimed to aggregate 2D perovskite data from diverse experiments and literature sources, as evidenced by the work of Marchenko et al.~\cite{marchenko2020database}. A remarkable facet of this undertaking is the establishment of an open-access repository teeming with crystal structures featuring atomic coordinates. Moreover, the database offers both density functional theory (DFT)-derived and experiment-based (Exp) bandgap values, expressed in electronvolts (eV), alongside machine learning (ML)-based atomic partial charges. This repository, meticulously maintained by the NMSE, stands unrivaled in its comprehensiveness, serving as an indispensable resource for delving into the world of 2D perovskite crystalline materials. Currently, the repository encompasses an impressive collection of 849 compounds possessing 3D structures. Utilizing the pymatgen package~\cite{ong2013python} to process CIFs~\cite{hall1991crystallographic}, we have gathered a total of 716 compounds characterized by DFT-based bandgap values or labels. Additionally, we have included 235 compounds with experimentally determined bandgap values.

\begin{figure}%[tbhp]
\centering
\includegraphics[width=0.5\linewidth]{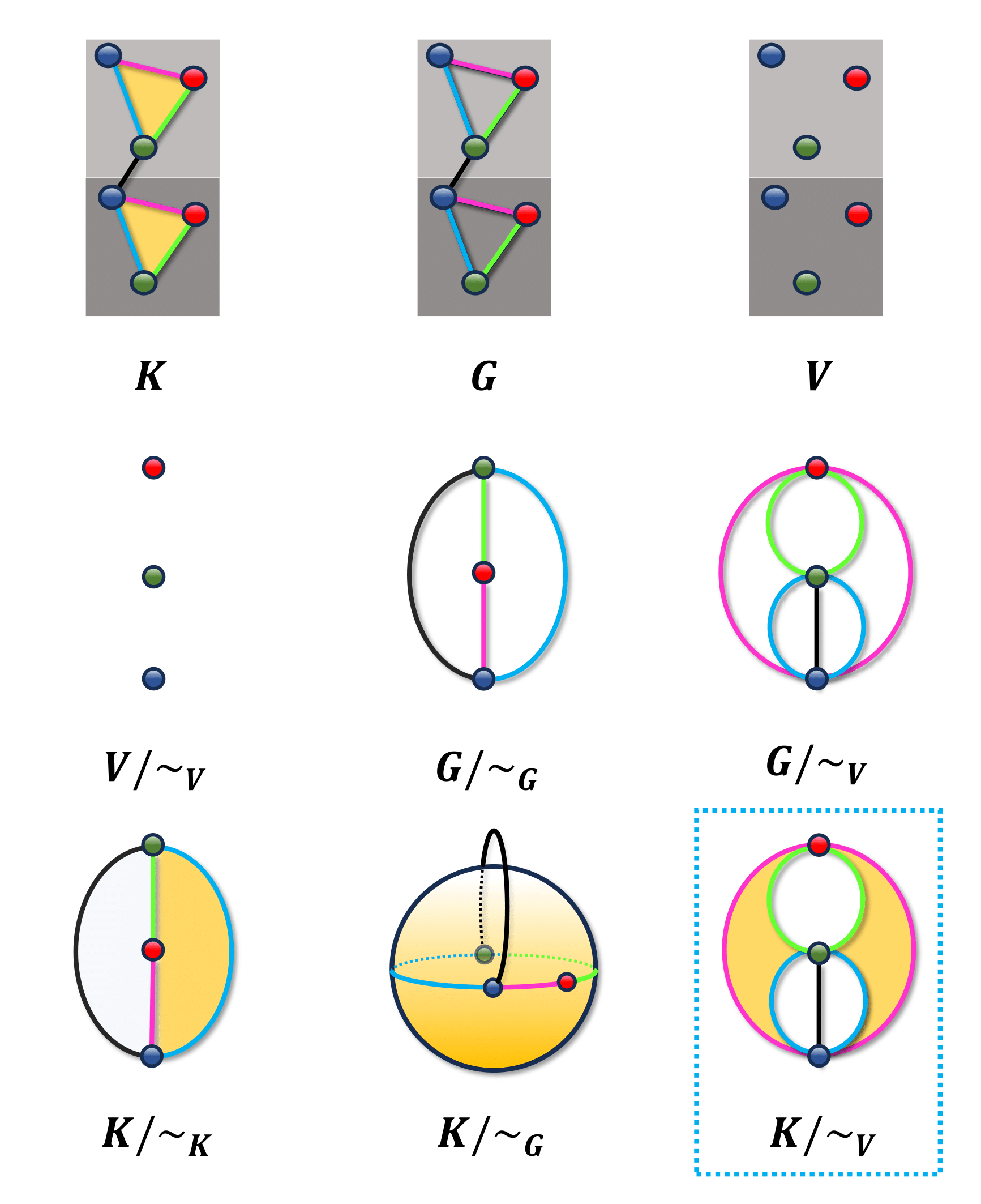}
\caption{Illustration of QCs obtained from a finite 2-periodic simplicial complex $K$ and its subcomplexes $G$ and $V$ within a 2-copies supercell embedded in $\mathbb{R}^2$. In this example, three equivalence relations are defined on $K$: $\sim_K$, $\sim_G$, and $\sim_V$, where $G$ and $V$ are subcomplexes of all 1-simplices (i.e., edges) and 0-simplices (i.e., vertices), respectively. For the illustrations of $K$, $G$, and $V$, vertices, edges, and triangles are annotated with the same color if they are equivalent based on the translation action on simplices. The QCs $V/\sim_V$, $G/\sim_G$, $G/\sim_V$, $K/\sim_K$, $K/\sim_G$, $K/\sim_V$ are illustrated in the second and third row, with particular emphasis on $K/\sim_V$ as the main focus in this study.}
\label{fig: Quotient Complexes}
\end{figure}

\subsubsection*{QC-based representation}

From a computational efficiency standpoint, this work utilizes a process for constructing the finite periodic motif, the QC filtration, and QC-based PH, aligning with the CIF data format of material structure. In terms of the finite periodic motif, our approach differs from supercell extensions, like the $3 \times 3 \times 3$ or $5 \times 5 \times 5$ extensions, as we concentrate on the union of $M$ and the other three translated motifs:
\begin{equation}
\label{Eq. 4M}
\begin{split}
V &= M \cup (M + v_1) \cup (M + v_2) \cup (M + v_3),
\end{split}
\end{equation}
where $\mathcal{B} = \{ v_1, v_2, v_3 \}$ is the $\mathbb{R}$-basis indicated by the material's CIF file. To concisely and efficiently represent the material, we primarily gather points in the original motif $M$ and incorporate the material's periodicity by considering its union with the translated copies $M + v_i$.

Once we have obtained the extended atomic system $V$ as shown in~\eqref{Eq. 4M}, we employ the bipartite distance $d: V \times V \rightarrow \mathbb{R} \cup \{ +\infty \}$ to characterize the interaction relationships between the original and shifted unit cells, with the aim of reducing computational costs:
\begin{equation}
\label{Eq. Distance function of extended cells-v1}
d(u,v) = \begin{cases}
 	 +\infty &\quad\null\text{ if } u \notin M \text{ and
 } v \notin M, \\
 	\Vert u - v \Vert_2 &\quad\null\text{ otherwise.}
 	\end{cases}
\end{equation}
The distance function in~\eqref{Eq. Distance function of extended cells-v1} associates a distance matrix, which can be used to construct the Vietoris--Rips filtration $K_{\epsilon_\bullet}$ of the atomic system as a point cloud~\cite{chazal2013interleaved}.  Similar bipartite distances, as in~\eqref{Eq. Distance function of extended cells-v1}, have been widely used for extracting PH-based representations for protein-protein, protein-ligand, and protein-nucleic acid interactions in computational bioinformatics~\cite{cang:2017integration, cang:2018representability}. Furthermore, by defining $v \sim_V w$ as $v - w \in \Lambda(\mathcal{B})$ (i.e., \eqref{Eq. The periodic equivalence relation}), the QC $\overline{K_{\epsilon_\bullet}}$ with $\overline{K_{\epsilon_i}} = K_{\epsilon_i}/\sim_V$ is defined.

Beyond enhancing computational efficiency, these configurations afford us a geometric insight into the relationship between unit cell information and the induced PH (refer to Theorem \ref{Theorem: Unit cell info as PB intervals}).

Fig. \ref{fig: Barcode information} shows an auxiliary example of the connection between unit cell information and QC-based PH. In this example, the unit cell is a parallelepiped generated by the basis $\{ v_1 = (10,0,0), v_2 = (0,20,0), v_3 = (0,0,30) \}$, containing $M = \{ (0,0,0), (5,10,15) \}$ as its motif. The figure depicts the lengths $|v_1|$, $|v_2|$, and $|v_3|$ appear as the birth values of certain persistence intervals within the ${\rm PB}_1^\infty$ of its QC filtration. Furthermore, more persistence intervals in ${\rm PB}_1^\infty$ have birth values beyond the lengths of the basis vectors and show their potential in extracting more complicated and intrinsic periodicity contained in the material.

\begin{figure}
	\centering
	\includegraphics[width=0.6\linewidth]{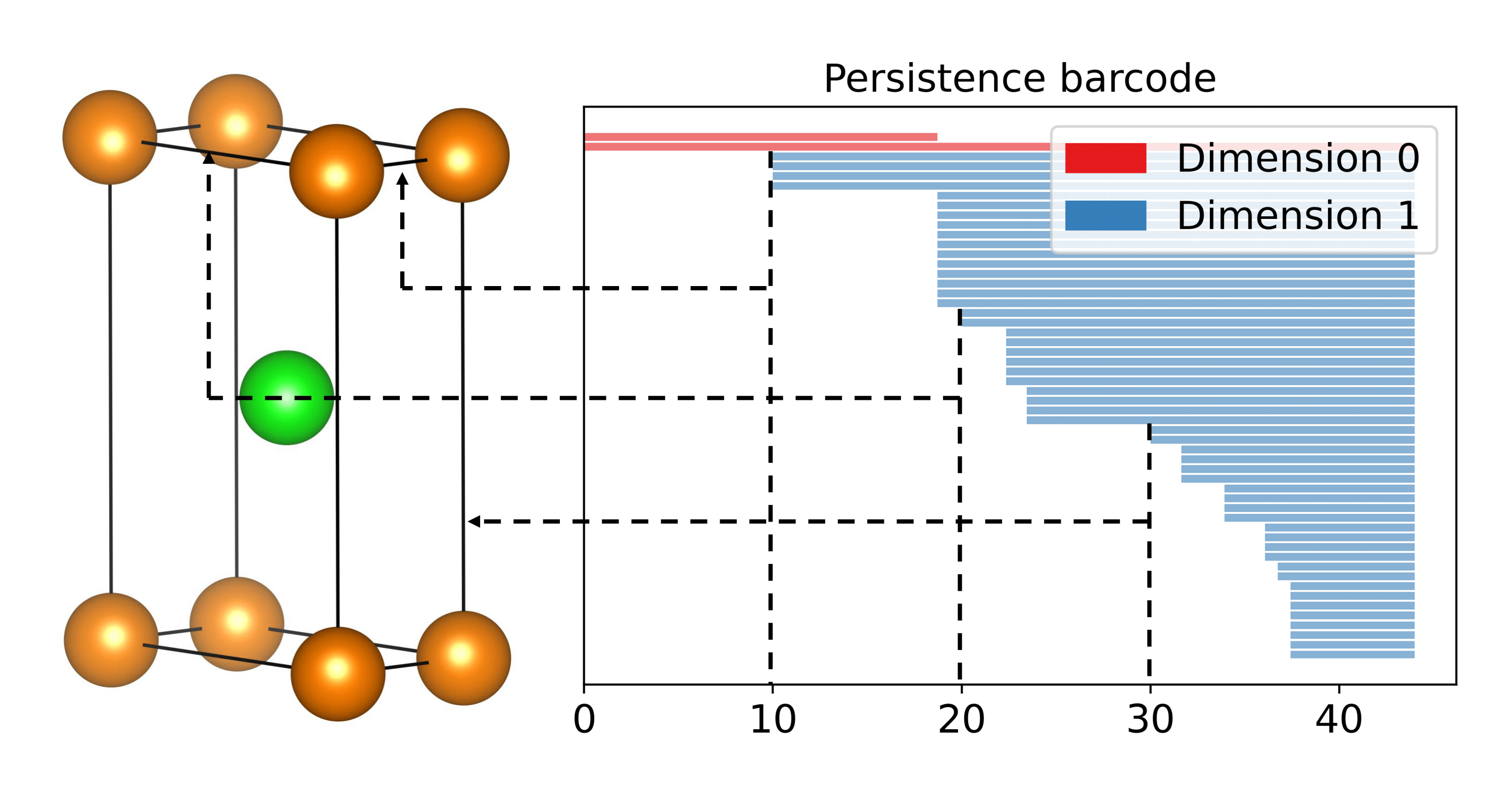}
	\caption{Illustration of a $10 \times 20 \times 30$ unit cell encompassing the motif ${ (0,0,0), (5,10,15) }$ and the induced persistence barcode of the quotient complex filtration. The quotient complex filtration is constructed from a Vietoris--Rips filtration based on settings in \eqref{Eq. The periodic equivalence relation}, \eqref{Eq. 4M}, and \eqref{Eq. Distance function of extended cells-v1} with filtration levels ranging from $0$ to $40$. }	
	\label{fig: Barcode information}
\end{figure}

\subsubsection*{QC-based descriptors}

To generate descriptors and features from QC-based filtrations, we compile birth and death information within the persistence barcodes ${\rm PB}_0$, ${\rm PB}_1^{\rm finite}$, ${\rm PB}_1^{\infty}$, and ${\rm PB}_2$ using statistical measures and Betti curves~\cite{umeda2017time}. In terms of statistical measurements, we analyze various aspects of birth, death, and lifespan information within these persistence barcodes. Specifically, we examine the maximum, minimum, median, quartiles, mean, and standard deviation for both birth and death values. Notably, for ${\rm PB}_0$ and ${\rm PB}_1^{\infty}$ with fixed birth and infinite death values, our focus centers on death and birth values, including normalized birth and death values, i.e.,
\begin{equation*}
\frac{d}{\sum_{(0,d') \in {\rm PB}_0(\overline{K_\bullet})} d'}  \text{ \ and \ } \frac{b}{\sum_{(b',\infty) \in {\rm PB}_1^{\infty}(\overline{K_\bullet})} b'}.
\end{equation*}

On the other hand, we calculate both normalized and non-normalized Betti curves as features derived from the barcodes ${\rm PB}_0$, ${\rm PB}_1^{\rm finite}$, ${\rm PB}_1^{\infty}$, and ${\rm PB}_2$~\cite{umeda2017time}. These curves capture the $q$-Betti number ($q = 0,1,2$) of the corresponding barcode at various filtration levels within the range $[0,T]$, where $T$ is the maximal filtration level of the PH. In application, $T$ is set to be $10$ (\r{A}) for computation efficiency. For further details on the implementation, refer to Appendix \ref{Appendix: Quotient Complex Descriptors (QCDs)} and the technique document of our code released on GitHub.

Additionally, as part of our specialized QCDs as shown in Theorem S8 and Fig. \ref{fig: Barcode information}, we further enhance the feature set by considering $(|v_1|, |v_2|, |v_3|, |v_1 + v_2|, |v_1 + v_3|, |v_2 + v_3|, |v_1 + v_2 + v_3|)$ with $|v_1| \leq |v_2| \leq |v_3|$ to emphasize the periodic information within the material. This feature depicts the edge length information of the unit cell representation and has been applied in GNN models as a crucial feature for crystalline structures~\cite{yan2022periodic}.

%\paragraph*{Element-specific QCDs}
Element-specific QCDs pertain to the geometric and topological characteristics of specific atomic systems or their combinations. Many other types of descriptors have demonstrated impressive efficacy in various applications, including bioinformatics molecules~\cite{cang:2017integration, szocinski2021awegnn} and perovskite data~\cite{mayr2021global, anand2022topological}, showcasing their high performance.

Due to this advantage, our approach deviates from capturing the entire structure of a 2D perovskite. Instead, we leverage the QCDs, which are derived from diverse atomic site combinations and types (i.e., atom sets). Specifically, the atomic sites encompass the B-site ($\rm B$), X-site ($\rm X$), and their combinations (e.g., ${\rm A}_{\rm C}{\rm B}{\rm X}$). Simultaneously, our consideration extends to QC descriptors affiliated with distinct atom types, spanning from $\rm O$ and $\rm N$ (associated with the A-site atom type) to $\rm Bi$, $\rm Cd$, $\rm Pb$, and $\rm Sn$, etc. (pertaining to the B-site atom type), and finally $\rm Cl$, $\rm Br$, and $\rm I$ (linked to the X-site atom type). These particular QCDs, originating from designated atom types or sites, are referred to as element-specific QCDs (cf. Table S1).

\subsubsection*{QC-GBT for 2D perovskite design}

For the hyperparameters of our GBT model, we employed 10,000 estimators with a maximum depth of 7 to analyze the input DF features. The loss metric chosen was the squared error. The minimum number of samples was set at 2, the learning rate at 0.001, and the subsampling rate during training was established at 0.7.

The scale of the training database significantly impacts the performance of AI models. Given the continuous updates to the NMSE database, the models listed in Table \ref{Table: Performances} draw from varying data sizes. Specifically, the MLM1 model leverages 515 compounds~\cite{marchenko2020database}, the SOAP-based KRR incorporates 445 compounds~\cite{mayr2021global}, while all GNN models (GCN \cite{welling2016semi}, ECCN \cite{park2020developing}, CGCNN \cite{xie2018crystal}, TFGNN \cite{shi2020masked}, SIGNNA, and SIGNNA$_c$ \cite{na2023substructure}) rely on 624 compounds. To ensure equitable comparisons with established state-of-the-art models~\cite{marchenko2020database,xie2018crystal,mayr2021global}, our QC-GBT models were also trained on the same datasets as their predecessors. Table \ref{Table: Performances} showcases the performance of QC-GBT models across various data sizes. Evidently, our QC-GBT consistently outperforms most of the prior models on their respective datasets.

Consistently, based on the MAE value (eV), our models emerge as the top-performing choice across various training sub-datasets within the NMSE database. Notably, as the NMSE database is continuously updated, the models outlined in Table \ref{Table: Performances} undergo training using diverse perovskite data volumes. It's important to observe that the dataset size for each model is outlined in Table \ref{Table: Performances}, typically falling into four categories: 445, 515, 624, and 716 instances. To ensure equitable comparison, we've equally trained our QC-based models using these four dataset sizes, with outcomes showcased in Table \ref{Table: Performances}. We observe that QC-GBT performs best on the dataset with 716 compounds in terms of COD and MAE, which aligns with machine learning theory. However, although QC-GBT performs slightly worse on smaller datasets, the average PCC, MAE, and RMSE are similar to the best model. Particularly, when QC-GBT is applied to datasets with 445 compounds, the average MAE is still superior to the baseline models. This indicates that the proposed QC-GBT demonstrates stable performance with respect to dataset size, showcasing the ability of QCDs to effectively capture the geometric and topological structure of perovskites and enhance the prediction of their bandgaps.

\subsection*{Data and code availability}
All 2D perovskite data utilized in this work are available on the official website of the NMSE database:  \href{http://www.pdb.nmse-lab.ru/}{http://www.pdb.nmse-lab.ru/}.
The main code utilized for generating QC-based descriptors in this paper is available on the project's GitHub page: \href{https://github.com/peterbillhu/QCPH}{https://github.com/peterbillhu/QCPH}.

\section*{Acknowledgement}
This work was supported in part by the Nanyang Technological University SPMS Collaborative Research Award 2022, Singapore Ministry of Education Academic Research fund Tier 2 grants MOE-T2EP20120-0013 and MOE-T2EP20221-0003, as well as the National Research Foundation (NRF), Singapore under its NRF Investigatorship (NRF-NRFI2018-04) and Competitive Research Program (CRP) (NRF-CRP25-2020-0004).

\bibliographystyle{unsrt}
\bibliography{references}  %%% Uncomment this line and comment out the ``thebibliography'' section below to use the external .bib file (using bibtex) .

\appendix
\section{Mathematical Foundation of QC Model}
\label{Appendix: Mathematical Foundation of QC Model}
In this section, we establish the mathematical foundation for the quotient complex (QC)-base framework utilized in this work. Specifically, in Section \ref{SI: Topological Background}, we introduce the topological background of quotient complexes, including materials and tools such as simplicial complexes, quotient topological spaces, and adjunction spaces that are applied to establish the quotient complexes from simplicial complexes. All the materials and theories mentioned in the section can be found in~\cite{munkres2000topology, munkres2018elements, brown2006topology, vick2012homology}. For the self-contained purpose, although all of the theorems in the section can be found directly or have similar statements in the references, we provide some of the proofs of the theorems that have a specific form used in this work. 

Section \ref{SI: Quotient Complexes} introduces the proposed construction of QCs from a simplicial complex embedded in a Euclidean space. In this work, our primary focus lies in constructing the QC $\overline{K}$ from the embedded simplicial complex $K$ by gluing the 0-simplices together. Additionally, based on the adjunction space construction introduced in Section \ref{SI: Topological Background}, we present a simplicial complex representation $\widetilde{K}$ of the established QC. Notably, the complexes $\overline{K}$ and $\widetilde{K}$ are homotopy equivalent, and hence, they yield the same homology information. In particular, leveraging the structure of the proposed simplicial complex $\widetilde{K}$, we further explore the relationships between the homology groups $H_q(\overline{K})$ and $H_q(K)$ (Theorem \ref{Theorem: onto, one-one, and iso}).

In Section \ref{SI: QC based PH}, we present the construction of the proposed QC filtrations and their persistent homology (PH). We will see that the $\widetilde{K}$ representation proposed in Section \ref{SI: Quotient Complexes} offers a combinatorial perspective to the study of persistent homology within quotient complexes, thereby facilitating the construction of the PH of $\widetilde{K}$. Especially, based on the homotopy equivalent representation $\widetilde{K}$ for $\overline{K}$, we also show that investigating the PH of $\overline{K_\bullet}$ is equivalent to the study on the PH of $\widetilde{K_\bullet}$.

Lastly, in Section \ref{SI: Persistence Barcode Analysis}, we explore the persistence barcodes (PBs) associated with both the simplicial complex filtration and the generated QC filtration. In particular, we delve into the 0-th and 1-th PBs of $K$ and $\widetilde{K}$ and represent the main theorems of the work (Corollary \ref{Corollary: Pre-thm of 0-ladder}, \ref{Coro: inclusion relation of the first PBs} ) 

\subsection{Topological Background}
\label{SI: Topological Background}

\paragraph{Simplicial Complexes}
An \textit{abstract simplicial complex} $\mathcal{K}$ over a set $V$ is a collection of non-empty finite subsets of $V$ with the following property: if $\sigma \in \mathcal{K}$ and $\tau \subseteq \sigma$, then $\tau \in \mathcal{K}$.  The set $V$ is called the \textit{vertex set} of $\mathcal{K}$. Any element $\sigma \in \mathcal{K}$ consisting of $n+1$ elements (called \textit{vertices}) is termed an $n$\textit{-simplex}, where $n$ is also defined as the \textit{dimension} of $\sigma$ (denoted as $\dim(\tau) = n$).  If $\sigma \in \mathcal{K}$ and $\tau \subseteq \sigma$, then $\tau$ is called a face of $\sigma$. Especially, $\tau$ is called an $n$\textit{-dimensional face} of $\sigma$ if $\tau \subseteq \sigma$ and $\dim(\tau) = n$. An abstract simplicial complex is called \textit{finite} if it is finite as a set i.e., it contains only finitely many simplices. In particular, any abstract simplicial complex over a finite vertex set must be a finite simplicial complex. 

In this study, we mainly consider \textit{simplicial complexes} that are the \textit{geometric realization} of an abstract simplicial complex embedded in a Euclidean space $\mathbb{R}^d$.  Especially, every element $v$ in $V$ is embedded as a vector $\iota(v)$ in $\mathbb{R}^d$ such that every $n$-simplex $\sigma = \{ v_0, v_1, ..., v_n \}$ ($n \leq d$) in $\mathcal{K}$ corresponds to the convex hull of $\{ \iota(v_0), \iota(v_1), ..., \iota(v_n) \}$, where $\iota(v_0), \iota(v_1), ..., \iota(v_n)$ are affinely independent in $\mathbb{R}^d$. One important fact is that every abstract simplicial complex must have a geometric realization~\cite{munkres2018elements}. In this paper, we consider an abstract simplicial complex as the foundational structure of a simplicial complex in a Euclidean space, frequently using these two concepts interchangeably in various contexts. To avoid unnecessary distinction, we denote both by $K$ when no confusion arises.

\paragraph{Simplicial Homology} 
The (singular) homology group, a fundamental tool in algebraic topology~\cite{vick2012homology,munkres2018elements}, serves as a homotopy invariant and is more computable relative to homotopy groups. This group, or vector space, effectively detects various dimensional hole structures within a topological space, including 0-dimensional connected components, 1-dimensional loops, and higher-dimensional voids. These detections are invariant under continuous deformations of the spaces, forming a robust basis for topological data analysis (TDA).

As a specific instance of singular homology, \textit{simplicial homology} is tailored for computing the homological information of simplicial complexes. In contrast to singular homology, simplicial homology places greater emphasis on face relations of simplices in the abstract simplicial complex structure, providing more intuitive and computable boundary maps. Mathematically, for a simplicial complex $K$ over a vertex set $V$ and a non-negative integer $q$, the $q$-\textit{chain space} over the binary field $\mathbb{Z}_2 = \mathbb{Z}/2\mathbb{Z}$ is defined as the $\mathbb{Z}_2$-vector space
\begin{equation*}
C_q(K) = C_q(K;\mathbb{Z}_2) = \bigoplus_{\sigma \in K_q} \mathbb{Z}_2 \sigma, 
\end{equation*}
where $K_q$ denotes the set of all $q$-simplices of $K$. We also set $C_q(K) = 0$ for $q < 0$ as a convention. Furthermore, to define the \textit{homology} of the simplicial complex, the $q$\textit{-th boundary map} $\partial_q: C_q(K) \rightarrow C_{q-1}(K)$ is defined as the $\mathbb{Z}_2$-linear map that extends the mapping
\begin{equation*}
\sigma = [v_0, v_1, ..., v_q] \longmapsto \sum_{i = 0}^q \ (-1)^i [v_0, ..., \widehat{v_i}, ..., v_q] = \sum_{i = 0}^q \ [v_0, ..., \widehat{v_i}, ..., v_q], 
\end{equation*}
where $[v_0, ..., \widehat{v_i}, ..., v_q]$ denotes the $(q-1)$-simplex $[v_0, ..., v_{i-1}, v_i, ..., v_q]$, and the equality holds since the based field is $\mathbb{Z}_2$. Especially, $\partial_{q-1} \circ \partial_q = 0$ for every $q \geq 1$~\cite{munkres2018elements,vick2012homology}, and the $q$\textit{-th homology} (over $\mathbb{Z}_2$) is defined as the quotient vector space
\begin{equation*}
H_q(K) = H_q(K;\mathbb{Z}_2) = \frac{\ker(\partial_q)}{{\rm im}(\partial_{q+1})}.    
\end{equation*}

The simplicial homology detects the information hole structures within the simplicial complex, including connected components, loops, and higher-dimensional voids. In particular, the $q$\textit{-th Betti number} of a simplicial complex $K$ is defined as the dimension of the $q$-th homology of $K$, which counts the number of $q$-dimensional holes within $K$ and is denoted by $\beta_q(K) = {\rm dim}_{\mathbb{Z}_2} H_q(K)$.  

In addition, suppose $\iota: K \hookrightarrow L$ is an inclusion of simplicial complexes, then $C_q(K) \subseteq C_q(L)$ for each $q$. Moreover, every rectangle of the ladder
\begin{equation*}
\xymatrix@+1.0em{
                \cdots
                \ar[r]^{}
                & C_{q+1}(K)
                \ar[r]^{\partial^K_{q+1}}
                \ar@{^{(}->}[d]^{}
                & C_q(K)
                \ar[r]^{\partial^K_{q}}
                \ar@{^{(}->}[d]^{}
                & C_{q-1}(K)
                \ar[r]^{}
                \ar@{^{(}->}[d]^{}
                & \cdots
                \\
                \cdots
                \ar[r]^{}
                & C_{q+1}(L)
                \ar[r]^{\partial^L_{q+1}}
        	& C_{q}(L)
			\ar[r]^{\partial^L_{q}}
			& C_{q-1}(L)
                \ar[r]^{}
                & \cdots
                }
\end{equation*}
of chain spaces and boundary maps is commutative. This implies that the mapping $H_q(K) \rightarrow H_q(L)$ defined by $z + {\rm im}(\partial^X_{q+1}) \mapsto z + {\rm im}(\partial^Y_{q+1})$ with $z \in \ker(\partial^X_{q})$ is a well-defined $\mathbb{Z}_2$-linear transformation.

\begin{figure}
	\centering
	\includegraphics[width=1.0\linewidth]{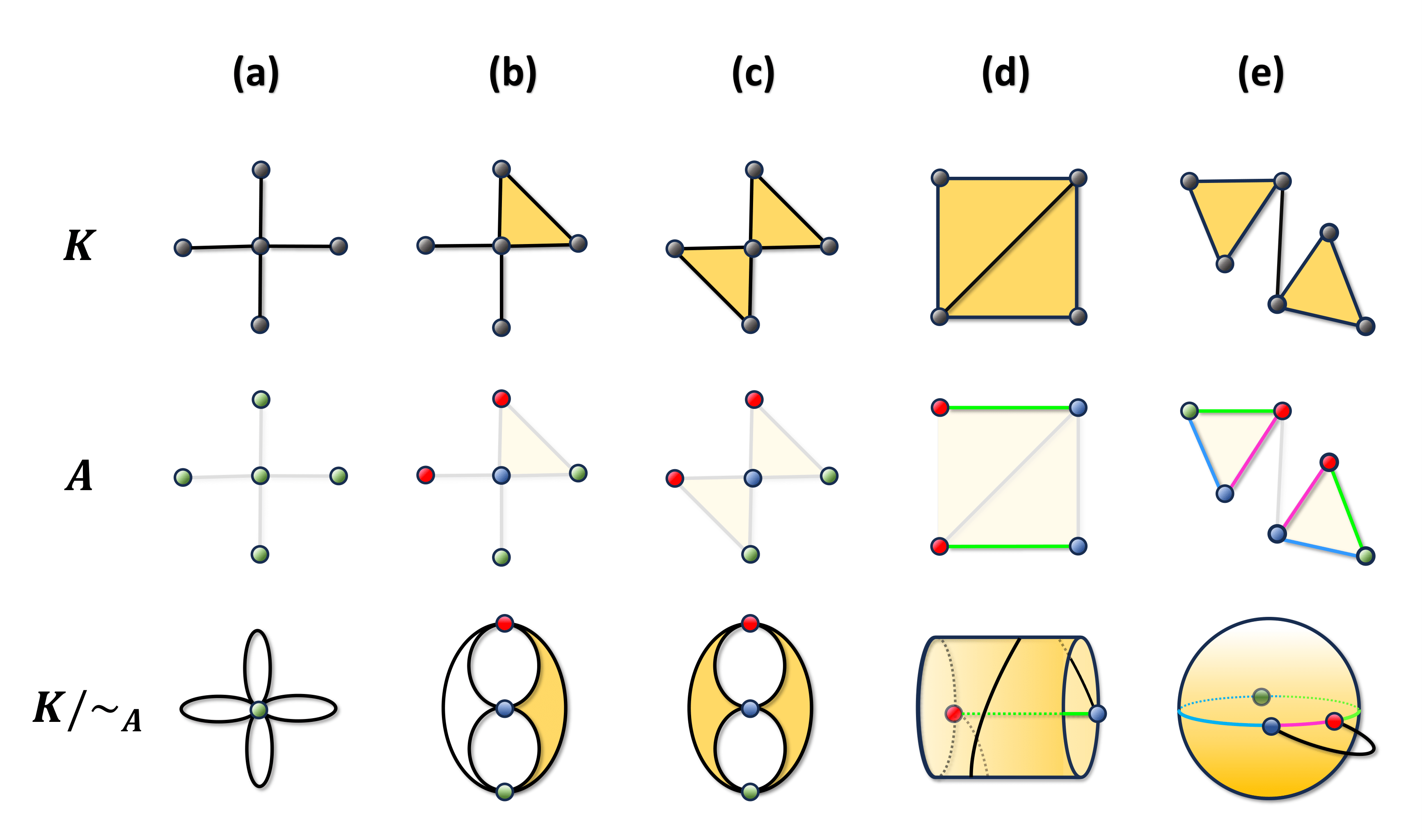}
	\caption{Illustration of simplicial complexes $K$ in dimension $\leq 2$ and their QCs $K/\sim_A$, where every $\sim_A$ is an equivalence relation defined on a subcomplex $A$ of $K$. Simplices (vertices or edges) in $A$ are deemed equivalent if they share the same color, thereby determining an equivalence relation $\sim_A$ on $A$. By gluing the same-colored vertices and edges in $K$ together, the associated QCs $K/\sim_A$ are depicted in the third row. The Betti triples $(\beta_0, \beta_1, \beta_2)$ of the QCs $K/\sim_A$ in \textbf{(a)}, \textbf{(b)}, \textbf{(c)}, \textbf{(d)}, and \textbf{(e)} are $(1,4,0)$, $(1,3,0)$, $(1,2,0)$, $(1,1,0)$, and $(1,1,1)$.}	
	\label{fig: quotient spaces and quotient complexes}
\end{figure}

\paragraph{Quotient Topological Spaces} 
In the realm of topology, the quotient operation is a typical method for creating a new topological space from a given space by gluing its subspaces together. Examples include the M\"{o}bius band, torus, and Klein bottle, all of which can be formed by connecting corresponding oriented edges of a single rectangle~\cite{munkres2000topology}. Formally speaking, given a topological space $X$ and an equivalence relation $\sim$ on $X$, all the equivalence classes $[x]$ with $x \in X$ form a set $X/\sim$. Especially, one can equip the finest topology on $X/\sim$ that ensures the canonical map $q: X \rightarrow X/\sim$, defined by $x \mapsto [x]$, to be continuous. The map $q$ is called the quotient map, which satisfies the following \textit{universal property}: for every continuous $f: X \rightarrow Y$ satisfying $f(x) = f(x')$ whenever $x \sim x'$, there is a unique continuous map $\overline{f}: X/\sim \rightarrow Y$ such that $\overline{f} \circ q = f$. This universal property can be illustrated through the use of a diagram, as depicted below.
\begin{equation*}
\xymatrix@+1.5em{
                & X
                \ar[r]^{q}
                \ar[rd]_{f}
                & X/\sim
                \ar[d]^{\exists! \ \overline{f}}
                \\
                & 
        	& Y
			}    
\end{equation*}

We employ Fig. \ref{fig: quotient spaces and quotient complexes} to elucidate examples of quotient topological spaces. In these examples, the based topological spaces $X = K$ are simplicial complexes embedded in $\mathbb{R}^2$. By defining equivalence relations $\sim$ on $X$, where $\sim \ = \ \sim_A$ is related to a subspace $A$ of $X$, the corresponding quotient topological spaces $X/\sim_A$, called \textit{quotient complexes} (see Section \ref{SI: Quotient Complexes}), are depicted in the third row of the figure.

\paragraph{Adjunction Spaces}
Let $X$ be a topological space, and $A$ be a subspace of $X$. For a continuous map $f: A \rightarrow Y$, one can define the least equivalence relation on the disjoint union space $X \sqcup Y$ generated by $a \sim f(a)$ for all $a \in A \subseteq X$. The resulting \textit{adjunction space}, denoted by $X \cup_f Y$, is defined as the quotient topological space $X \sqcup Y / \sim$. This quotient operation serves as the foundation for constructing cell complexes and proves to be a powerful tool for exploring the homology theory of topological manifolds~\cite{vick2012homology}. In particular, within this work, adjunction spaces provide a convenient framework for studying the structure of quotient complexes.

Let $X$ and $Y$ be topological spaces, $A$ a subspace of $X$, and $f: A \rightarrow Y$ a continuous map. Then there is a canonical equivalence relation $\sim_f$ on $X$, defined by $x \sim_f x'$ if and only if $x = x'$ or $f(x) = f(x')$ with $x, x' \in A$. Because $f(x) = f(x')$ whenever $x \sim_f x'$, the universal property of topological spaces leads to a continuous map
\begin{equation}
\label{Equation: Universal property map-1}
\phi: X /\sim_f \rightarrow X \cup_f Y    
\end{equation}
that assigns every equivalence class of an $x$ in $X/\sim_f$ to the equivalence class of $x$ in $X \cup_f Y$. The following theorem lists some fundamental properties of the continuous map $\phi$.

\begin{theorem}
\label{Theorem: Adjunction space versus Gluing space}
Let $X$ and $Y$ be topological spaces, $A$ a subspace of $X$, and $f: A \rightarrow Y$ be a continuous map.  Let $\phi$ be the continuous map defined in~\eqref{Equation: Universal property map-1}. Then, the following assertions hold.
\begin{itemize}
\item[\rm (a)] $\phi$ is a one-to-one map;
\item[\rm (b)] $\phi$ is onto if $f$ is onto;
\item[\rm (c)] For any subset $\mathcal{S}$ of $X/\sim$, $q^{-1}(\phi(\mathcal{S})) \cap X = i_X^{-1}(q^{-1}(\phi(\mathcal{S})))$ and $q^{-1}(\phi(\mathcal{S})) \cap Y = f(p^{-1}(\mathcal{S}) \cap A)$.
\end{itemize}
\end{theorem}
\begin{proof}
Let $p: X \rightarrow X/\sim_f$ and $q: X \sqcup Y \rightarrow X \cup_f Y$ be the quotient map, and let $i_X: X \hookrightarrow X \sqcup Y$ be the inclusion map. By the universal property of quotient topological spaces, the diagram 
\begin{equation*}
\xymatrix@+1.5em{
                & X
                \ar@{^{(}->}[r]^{i_X}
                \ar[d]_{p}
                & X \sqcup Y
                \ar[d]^{q}
                \\
                & X/\sim_f
                \ar[r]_{\phi}
        	& X \cup_f Y
			}    
\end{equation*}
commutes. For (a), let $x$ and $x'$ be two elements in $X$ that represent equivalence $[x]$ and $[x']$ in $X/\sim_f$. If $x$ and $x'$ define the same equivalence class in $X \cup_f Y$, then either $x = x'$ or $f(x) = f(x')$ with $x, x' \in A$. This shows that $x \sim x'$ in $X$ and proves that $\phi$ must be one-to-one. For $(b)$, if $f$ is onto, then every $y$ in $Y$ is equivalent to an element $x$ in $X$, and this shows that $\phi$ must be onto. 

For (c),  let $\mathcal{S}$ be a subset of $X/\sim_f$, then $x \in q^{-1}(\phi(\mathcal{S})) \cap X$ if and only if $x \in i_X^{-1}(q^{-1}(\phi(\mathcal{S})))$. For the second equality, if $y \in q^{-1}(\phi(\mathcal{S})) \cap Y$, then $q(y) = \phi(p(x)) = q(x)$ for some $x \in X$ with $p(x) \in \mathcal{S}$. Especially, $x$ and $y$ define the same equivalence class in $X \cup_f Y$. By the definition of $X \cup_f Y$, we have $x \in A$ and $y = f(x)$. Thus, $y \in f(p^{-1}(\mathcal{S}) \cap A)$. Conversely, $y \in f(p^{-1}(\mathcal{S}) \cap A)$ implies $y = f(x)$ for some $x \in p^{-1}(\mathcal{S}) \cap A$. Then, $q(y) = q(f(x)) = q(x) = \phi(p(x)) \in \phi(\mathcal{S})$.
\end{proof}

In the study, we utilize Theorem \ref{Theorem: Adjunction space versus Gluing space} in the following scenario. Suppose $\sim_A$ is an equivalence relation on $A$, then $\sim_A$ can be also viewed as an equivalence relation on $X$. Especially, for $x, x' \in X \setminus A$, $x \sim_A x'$ only if $x = x'$. Let $Y = A/\sim_A$ be the quotient space with the quotient map $f: A \rightarrow Y$, then the adjunction space $X \cup_f Y$ is defined. Especially, if $x, x' \in A$, then $x \sim_f x'$ if and only if $f(x) = f(x')$, which is equivalent to $x \sim_A x'$. Then the map in~\eqref{Equation: Universal property map-2} becomes 
\begin{equation}
\label{Equation: Universal property map-2}
\phi: X /\sim_A \rightarrow X \cup_f Y.    
\end{equation}
Because the quotient map $f: A \rightarrow Y$ is onto, the induced map $\phi$ is bijective (Theorem \ref{Theorem: Adjunction space versus Gluing space}(b)).  Furthermore, with additional assumptions, the bijective map is elevated to a homeomorphism. For instance, as outlined in the following corollary, when $A$ is closed and $f$ is a closed map, then $\phi$ constitutes a homeomorphism.

\begin{corollary}
\label{Coro: Adjunction space versus Gluing space-v2}
Let $X$, $Y$, $A$, $f: A \rightarrow Y$ and $\phi: X/\sim_f \rightarrow X \cup_f Y$ be defined as in~\eqref{Equation: Universal property map-1}. Suppose $f$ is onto, $A$ is closed in $X$, and $f$ is a closed map, then $\phi$ is a homeomorphism. 
\end{corollary}
\begin{proof}
Let $\mathcal{C}$ be an arbitrary closed subset of $X/\sim$. Because $f$ is onto, map $\phi$ is bijective by Theorem \ref{Theorem: Adjunction space versus Gluing space}(a) and (b). By Theorem \ref{Theorem: Adjunction space versus Gluing space}(c), $q^{-1}(\phi(\mathcal{C})) \cap X = i_X^{-1}(q^{-1}(\phi(\mathcal{C}))) = (p^{-1}((\phi^{-1} \circ \phi)(\mathcal{C}))) = p^{-1}(\mathcal{C})$, which is closed in $X$. On the other hand, we have $q^{-1}(\phi(\mathcal{C})) \cap Y = f(p^{-1}(\mathcal{C}) \cap A)$. Because $A$ is closed in $X$ and $f$ is closed, the set $q^{-1}(\phi(\mathcal{C})) \cap Y$ is closed in $Y$. Because $q$ is the quotient map, $\phi(\mathcal{C})$ is closed. Thus, $\phi$ is a bijective, continuous, and closed map, and hence a homeomorphism. 
\end{proof}

Alternatively, as indicated in the following corollary, if $Y$ is a discrete space, then $\phi$ also forms a homeomorphism. In this work, this fact is naturally adapted to construct quotient complexes from a simplicial complex by gluing grouped vertices and is extensively utilized for analyzing the persistent homology of quotient complexes (see Section \ref{SI: Quotient Complexes} and \ref{SI: QC based PH}).

\begin{corollary}
\label{Coro: Adjunction space versus Gluing space-v1}
Let $X$, $Y$, $A$, $f: A \rightarrow Y$ and $\phi: X/\sim_f \rightarrow X \cup_f Y$ be defined as in~\eqref{Equation: Universal property map-1}. Suppose $f$ is onto and $Y$ is discrete, then $\phi$ is a homeomorphism. 
\end{corollary}
\begin{proof}
Because $f$ is onto, map $\phi$ is bijective by Theorem \ref{Theorem: Adjunction space versus Gluing space}(a) and (b).  Let $\mathcal{U}$ be an arbitrary open subset of $X/\sim_f$.  Then
\begin{equation*}
q^{-1}(\phi(\mathcal{U})) \cap X = i_X^{-1}(q^{-1}(\phi(\mathcal{U}))) \text{ and } q^{-1}(\phi(\mathcal{U})) \cap Y = f(p^{-1}(\mathcal{U}) \cap A)    
\end{equation*}
by Theorem \ref{Theorem: Adjunction space versus Gluing space}(c). Because $Y$ is discrete, $q^{-1}(\phi(\mathcal{U})) \cap Y$ is open in $Y$. On the other hand, 
\begin{equation*}
q^{-1}(\phi(\mathcal{U})) \cap X = i_X^{-1}(q^{-1}(\phi(\mathcal{U}))) = p^{-1}(\phi^{-1}(\phi(\mathcal{U}))) = p^{-1}(\mathcal{U}), 
\end{equation*}
where the third equality holds since $\phi$ is bijective. In particular, $q^{-1}(\phi(\mathcal{U}))$ is open in $X \sqcup Y$. Because $q$ is the quotient map, $\phi(\mathcal{U})$ is open. In summary, $\phi$ is a bijective, continuous, and open map, and hence a homeomorphism.  
\end{proof}

For instance, if $\sim_A$ defines finitely many equivalence classes in $A$, and each equivalence class of $\sim_A$ constitutes a closed subset of $A$, then $A/\sim_A$ becomes a discrete space due to its finiteness and the closedness of every point. Notably, $X/\sim_A$ and $X \cup_f Y$ are homeomorphic. 

\subsection{Quotient Complexes (QCs)}
\label{SI: Quotient Complexes}

In this section, we present our proposed framework for constructing quotient complexes (QCs) derived from a simplicial complex embedded in the $d$-dimensional Euclidean space $\mathbb{R}^d$. Formally, let $K$ be a simplicial complex embedded in $\mathbb{R}^d$, $A$ be a subcomplex of $K$, and $\sim_A$ be an equivalence relation on $A$. Viewing $\sim_A$ as an equivalence on $K$, the \textit{quotient complex} (QC) $\overline{K}$ with respect to $\sim_A$ is defined as the quotient space $K/\sim_A$, which is a cell complex~\cite{hatcher2002algebraic}.

Especially, the main focus of this work is on the QC based on the equivalence relations of the $0$-simplices of the simplicial complex. In other words, we consider the QC of the form $K/\sim_V$, where $V$ is the vertex set of $K$ and $\sim_V$ is an equivalence relation on $V$.  Additionally, as a constraint on $V$ and $\sim_V$, we require that $V$ is a discrete set, and $\sim_V$ only defines finitely many equivalence classes that partition the set $V$. This constraint is compatible with the crystal structure, which extends a finite motif to the entire material, making the analysis of the homological structure more straightforward.

Fig. \ref{fig: quotient spaces and quotient complexes} illustrates simplicial complexes $K$ and the related QCs defined by various equivalence relations $\sim_A$ on subspaces $A$ of $K$. Specifically, spaces in the first row are simplicial complexes $K$ in dimension $\leq 2$. The second row depicts subspaces $A \subseteq K$ with colored vertices and edges. Simplices in $A$ with the same color determine an equivalence relation $\sim_A$. Especially, by gluing the same-colored vertices and edges in $K$, the associated QCs $K/\sim_A$ are depicted in the third row. Notably, $A$ in examples \textbf{(a)}, \textbf{(b)}, and \textbf{(c)} equal the vertex set $V \subseteq K$, which are the cases concerned in this paper. On the other hand, except for vertices, \textbf{(d)} and \textbf{(e)} consider $A$ as subcomplexes of $K$ containing $1$-simplices. 

\paragraph{Homotopy Equivalent Representaion}

In order to compute the quotient complex structure more efficiently, we introduce a homotopy equivalent approach that uses simplicial complexes $\widetilde{K}$ to represent the proposed quotient complex $\overline{K} = K/\sim_V$. Because homology is a homotopy invariant for topological spaces,  the homology spaces $H_q(\overline{K})$ and $H_q(\widetilde{K})$ are isomorphic. Especially, to elucidate the homotopy equivalence relation between the constructed $\widetilde{K}$ and the quotient complex $\overline{K}$, we refer to the following well-known theorem in homotopy theory.

\begin{theorem}[\cite{brown2006topology}]
\label{homotopy theorem in Brown's book}
Suppose the following diagram of topological spaces and continuous maps
\begin{equation*}
\xymatrix@+1.0em{
                & Y
                \ar[d]_{\varphi_Y}
                & A
                \ar[l]_{f}
                \ar@{^{(}->}[r]^{i}
                \ar[d]^{\varphi_A}
                & X
                \ar[d]^{\varphi_X}
                \\
                & Y'
        	& A'
				\ar[l]^{f'}
				\ar@{^{(}->}[r]_{i'}
				& X'
				}    
\end{equation*}
commutes, where $\varphi_X, \varphi_A, \varphi_Y$ are homotopy equivalences, and the inclusions $i, i'$ are closed cofibrations. Then the map $\varphi: X \cup_f Y \longrightarrow X' \cup_{f'} Y'$ induced by the maps $\varphi_X, \varphi_A, \varphi_Y$ is a homotopy equivalence.
\end{theorem}

By applying the result in Theorem \ref{homotopy theorem in Brown's book} to the specified settings of the simplicial complex $K \subseteq \mathbb{R}^d$ and the equivalence relation $\sim_V$, we can construct the corresponding quotient complex $\widetilde{K}$ as follows. Since $\sim_V$ is assumed to admit finitely equivalence classes in $V$, we set 
$V_1, ..., V_k$ to be the equivalence classes in $V$. To make the arguments in more computational aspects, we identify $\mathbb{R}^m$ as a subspace of $\mathbb{R}^{n}$ for $m < n$ via the canonical inclusion map $(x_1, ..., x_{m}) \mapsto (x_1, ..., x_{m}, 0, ..., 0)$. Especially, we let $z_1, z_2, ..., z_k$ be points in $\mathbb{R}^{d+k}$ such that $z_{j + 1} \in \mathbb{R}^{d+j+1} \setminus \mathbb{R}^{d+j}$ for $j = 0, ..., k-1$. In particular, we have $V = \bigsqcup_{j = 1}^k V_j$ and $Z := \{ z_1, z_2, ..., z_k \}$ is a discrete space.

By Proposition \ref{Coro: Adjunction space versus Gluing space-v1}, the quotient space $\overline{K} = K/\sim_V$ that emerges points in $V_i$ to a single point can be identified as the adjunction space $K \cup_f Z = K \cup_f \{ z_1, z_2, ..., z_k \}$, where $f: V \rightarrow Z$ is the local constant function with $f(V_j) = \{ z_j \}$. Furthermore, for each $j \in \{ 1, 2, ..., k \}$, we set
\begin{equation*}
S_j = \bigcup_{x \in V_j} [z_j, x] = \bigcup_{x \in V_j} \{ \lambda z_j + (1-\lambda) x \ | \ \lambda \in [0,1] \}, 
\end{equation*}
which is a star-shaped set, as it is the union of all line segments that connect each $z_j$ to the points in $V_j$. By construction, $S_j$ is a $1$-dimensional simplicial complex embedded in $\mathbb{R}^{d+j}$. In particular, $S_i \cap S_j = \emptyset$ for $i \neq j \in \{ 1,2, ..., k\}$ by the assumption of $z_1, z_2, ..., z_k$. Set $S = \bigsqcup_{j = 1}^k S_j$ and define $g: V \rightarrow S$ such that each $g|_{V_j}$ is the inclusion map $V_j \hookrightarrow S$. Finally, we set $\widetilde{K} = K \cup S = K \cup_g S$. Then $\widetilde{K}$ is a simplicial complex and the diagram 
\begin{equation*}
\xymatrix@+1.0em{
                & S
                \ar[d]_{r}
                & V
                \ar[l]_{g}
                \ar@{^{(}->}[r]^{i}
                \ar[d]^{{\rm id}_V}
                & K
                \ar[d]^{{\rm id}_K}
                \\
                & Z
        	  & V
				\ar[l]^{f}
				\ar@{^{(}->}[r]_{i}
				& K
				}    
\end{equation*}
commutes, where $r: S \rightarrow Z$ is defined such that $r|_{S_j}$ equals the constant map $S_j \rightarrow \{ z_j \} \hookrightarrow Z$. Since the continuous maps $r|_{S_j}$ can be shrunk synchronously, $r$ is a homotopy equivalence. Because any inclusion map of a subcomplex of a simplicial complex is a closed cofibration, $i$ is a closed cofibration (\cite{brown2006topology}, page 281). By Theorem \ref{homotopy theorem in Brown's book}, the following corollary follows.

\begin{corollary}
\label{Main corollary 1}
Let $K$, $V = \bigsqcup_{j = 1}^k V_j$, $K$,  $\overline{K}$, and $\widetilde{K}$ be defined as above. Then $\overline{K}$ and $\widetilde{K}$ are homotopy quivalent. In particular, the homology of $\overline{K}$ and $\widetilde{K}$ are isomorphic. 
\end{corollary}

In detail, by Theorem \ref{homotopy theorem in Brown's book}, the homotopy equivalence $r: S \rightarrow Z$ induces a homotopy equivalence $\widehat{r}$ from $\widetilde{K} = K \cup S = K \cup_g S$ to $ \overline{K} = K \cup_f Z$, which is defined as follows:
\begin{equation}
\label{Eq. Induced map from K cupg S to K cupf Z}
\widehat{r}(x) = \begin{cases}
\overline{x} &\quad\null\text{ if } x \in K \setminus S\\
\overline{r(x)} &\quad\null\text{ if } x \in S
\end{cases},   
\end{equation}
where the bar notation $\overline{\bullet}$ is applied to denote the equivalence classes in the quotient topological space $K \cup_f Z = \overline{K}$. 

\begin{definition}
Let $K$, $V = \bigsqcup_{j = 1}^k V_j$, $\overline{K}$, and $\widetilde{K}$ be defined as described above. The simplicial complex $S_j$ mentioned previously is referred to as the \textbf{gluing star} of the set $V_j$, and we designate the simplicial complex $S$ as the \textbf{gluing stars} of the set $A$. In particular, $V_j \subseteq S_j$ for each $j$ and $\widetilde{K} = K \cup S$.
\end{definition}

\begin{figure}
	\centering
		\centerline{\includegraphics[width=1\textwidth]{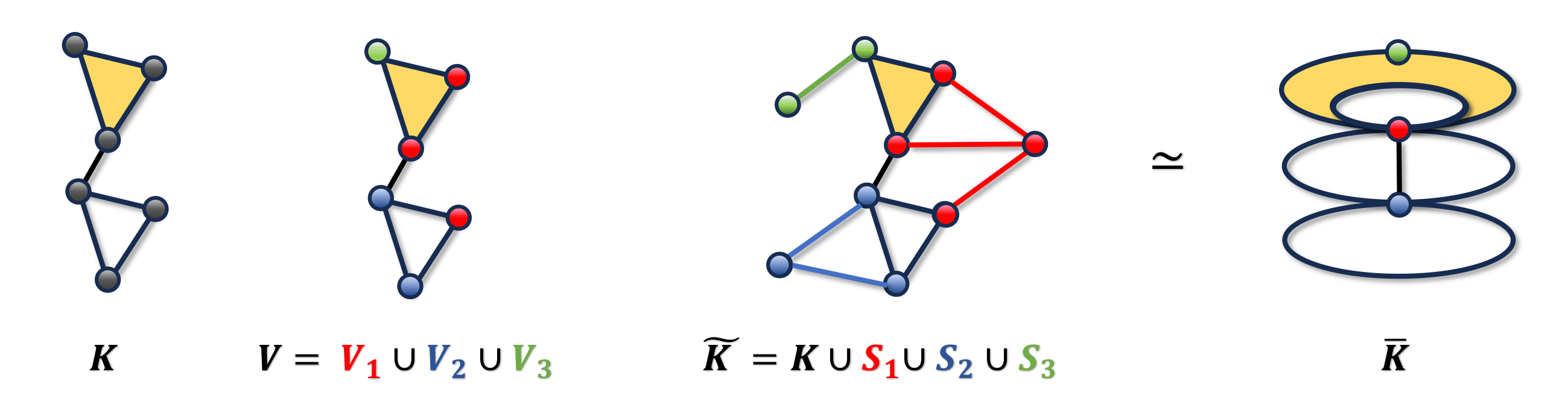}}
		\caption{An illustrative example showcases the construction of the simplicial complex $\widetilde{K}$, homotopy equivalent to the quotient complex $\overline{K}$. Here, $K$ represents a simplicial complex with 6 vertices, 7 edges, and 1 triangle. The vertex set $V$ is partitioned into three groups: $V_1$, $V_2$, and $V_3$. Introducing gluing stars $S_1$, $S_2$, and $S_3$ to the simplicial complex $K$ results in the formation of $\widetilde{K}$—a complex homotopy equivalent to the quotient complex $\overline{K}$.}
		\label{fig:homotopy_equivalence}
\end{figure}

Fig. \ref{fig:homotopy_equivalence} visually demonstrates the result of Corollary \ref{Main corollary 1}. Using the simplicial complex $K$ and the vertex set $V = V_1 \cup V_2 \cup V_3$, we construct the simplicial complex $\widetilde{K}$ as the union $K \cup S$, with $S = S_1 \cup S_2 \cup S_3$ denoting the gluing stars as illustrated in Fig. \ref{fig:homotopy_equivalence}. Following the implications of Corollary \ref{Main corollary 1}, both the simplicial complex $\widetilde{K}$ and the quotient space $\overline{K}$ share the same homotopy type. In particular, $\widetilde{K}$ and $\overline{K}$ have identical homological structures. Importantly, $K$ is a subcomplex of $\widetilde{K}$. Considering the homomorphism $\theta_\bullet: H_\bullet(K) \rightarrow H_\bullet(\widetilde{K})$, the following theorem holds.

\begin{theorem}
\label{Theorem: onto, one-one, and iso}
Let $V \subseteq \mathbb{R}^d$ be a discrete set and $K \subseteq \mathbb{R}^d$ be a simplicial complex over $V$. Let $\sim_V$ be an equivalence relation on $V$ with finitely many equivalence classes $V_1, V_2, ..., V_k$. Let $\widetilde{K}$ be defined as above. Then
\begin{itemize}
    \item[\rm (a)] $\theta_0: H_0(K) \rightarrow H_0(\widetilde{K})$ is onto;
    \item[\rm (b)] $\theta_1: H_1(K) \rightarrow H_1(\widetilde{K})$ is one-to-one;
    \item[\rm (c)] $\theta_q: H_q(K) \rightarrow H_q(\widetilde{K})$ is an isomorphism for $q > 1$.
\end{itemize}    
\end{theorem}
\begin{proof}
At the beginning of the proof, we consider the commutative ladder
\begin{equation*}
\xymatrix@+1.0em{
                \cdots
                \ar[r]^{}
                & C_3(K)
                \ar[r]^{\partial_3}
                \ar@{=}[d]^{}
                & C_2(K)
                \ar[r]^{\partial_2}
                \ar@{=}[d]^{}
                & C_1(K)
                \ar[r]^{\partial_1}
                \ar@{^{(}->}[d]^{}
                & C_0(K)
                \ar@{^{(}->}[d]^{}
                \ar[r]^{}
                & 0
                \\
                \cdots
                \ar[r]^{}
                & C_3(\widetilde{K})
                \ar[r]^{\partial_3}
        	& C_2(\widetilde{K})
				\ar[r]^{\partial_2}
				& C_1(\widetilde{K})
                \ar[r]^{\partial_1}
                & C_0(\widetilde{K})
                \ar[r]^{}
                & 0
				}
\end{equation*}
with chain complexes $C_\bullet(K)$ and $C_\bullet(\widetilde{K})$. By the construction of $\widetilde{K}$, we have $C_q(K) = C_q(\widetilde{K})$ for all $q \geq 2$ and $C_q(K) \subseteq C_q(\widetilde{K})$ for $q = 0, 1$. Especially, $\theta_q$ is an isomorphism whenever $q \geq 2$, hence (c) follows.

Let $\langle c \rangle$ be an equivalence class in $\ker(\theta_1)$ with $\partial_1(c) = 0$. Because $C_2(K) = C_2(\widetilde{K})$, and the class $\theta_1(\langle c \rangle)$ is also represented by $c$ in $H_1(\widetilde{K})$, $\langle c \rangle$ must be the zero element in $H_1(K)$, and this proves (b). 

Finally, let $\langle v \rangle_{\widetilde{K}}$ be an equivalence class in $H_0(\widetilde{K})$ with a $0$-simplex $v$ of $\widetilde{K}$. If $v \in K$, then $\theta_0(\langle v \rangle_K) = \langle v \rangle_{\widetilde{K}}$, where $\langle v \rangle_K$ is the equivalence class of $v$ in $H_0(K)$. For otherwise, $v = z$ for some $ z \in Z$. Especially, there is an $w \in V \subseteq K$ such that $[z, w]$ is a $1$-simplex in $\widetilde{K}$. Because $z$ and $w$ are homologous in $\widetilde{K}$, $\theta_0(\langle w \rangle_K) = \langle w \rangle_{\widetilde{K}} = \langle z \rangle_{\widetilde{K}}$. Because $H_0(\widetilde{K})$ is generated by all equivalence classes that are represented by vertices of $\widetilde{K}$, (a) follows.
\end{proof}

Theorem \ref{Theorem: onto, one-one, and iso}-(a) and (b) clarifies the relationships within the homology groups of the quotient complex $\overline{K} \simeq \widetilde{K}$ and the original simplicial complex $K$, specifically in dimensions $0$ and $1$. In contrast, Theorem \ref{Theorem: onto, one-one, and iso}-(c) asserts that the quotient complex $\overline{K} \simeq \widetilde{K}$ encapsulates identical homological information to that of $K$. This arises from the fact that the gluing operation only affects 0-dimensional simplices. Investigating the impact of gluing higher-dimensional simplices or subcomplexes presents a potential avenue for future research.

\subsection{QC-based Persistent Homology}
\label{SI: QC based PH}

\paragraph{Persistent Homology}
Beyond a single topological space and its homology, persistent homology (PH) directs attention to continuously changing topological spaces. In particular, one primary interest within the TDA community involves a sequence of included topological spaces $X_\bullet: \emptyset \subseteq X_1 \subseteq X_2 \subseteq \cdots \subseteq X_n$, referred to as a \textit{filtration} of topological spaces, particularly for the filtration of simplicial complexes. Based on the functoriality of homology,  this filtration induces the following sequence of $\mathbb{Z}_2$-vector spaces and $\mathbb{Z}_2$-linear transformations:
\begin{equation*}
{\rm PH}_q(X_\bullet): 0 \longrightarrow H_q(X_1) \xrightarrow{ \ \phi_1 \ } H_q(X_2) \xrightarrow{ \ \phi_2 \ } \cdots \xrightarrow{ \ \phi_{n-1} \ } H_q(X_n)   
\end{equation*}
for any given $q \geq 0$, where each $\phi_i$ is induced by the inclusion $X_i \hookrightarrow X_{i+1}$. 

The \textit{persistence barcode} (PB) serves as a comprehensive tool for summarizing persistent homology, capturing the birth and death information of local structures within the persistent homology sequence~\cite{Ghrist:2008, CZCG05}. More generally, the computation of birth and death information is applicable to any sequence of vector spaces and linear transformations. Mathematically, let
\begin{equation}
\label{Eq. Vector space sequence}
W_\bullet: 0 \xrightarrow{ \ \ \ \ } W_1 \xrightarrow{ \ \ \phi_1 \ \ } W_2 \xrightarrow{ \ \ \phi_2 \ \ } \cdots \xrightarrow{ \ \ \phi_{n-1} \ \ } W_n    
\end{equation}
be a sequence of vector spaces $W_i$ and linear maps $\phi_i$. Utilizing the composition maps $\phi_{j-1} \circ \cdots \circ \phi_{i+1} \circ \phi_i$ for $i < j \in { 1,2, ..., n }$, the evolution of elements in $W_\bullet$ defines their lifespan information and influences the resulting PB. 

We introduce the definition of PB by adhering to the following mathematical setting. The composition map $\phi_{j-1} \circ \cdots \circ \phi_{i+1} \circ \phi_i$ is denoted by $\phi_{i,j}$, and $\phi_{i,i}$ is defined as the identity function on $W_i$. An $s \in W_b$ with $b \in \{ 1, 2, ..., n \}$ is said to have a \textit{persistence interval} $(b, d)$ with \textit{birth} $b$ and \textit{death} $d \in \{ b+1, ..., n \} \cup \{ \infty \}$ if the following two properties hold: (a) $s \notin {\rm im}(\phi_{b-1})$ and (b) $\phi_{b,d'}(s) \notin {\rm im}(\phi_{b-1, d'})$ for $d' < d$ and $\phi_{b,d}(s) \in {\rm im}(\phi_{b-1, d})$. Moreover, $d = \infty$ if $s$ never dies within the sequence.

\begin{figure}
	\centering	\includegraphics[width=1\linewidth]{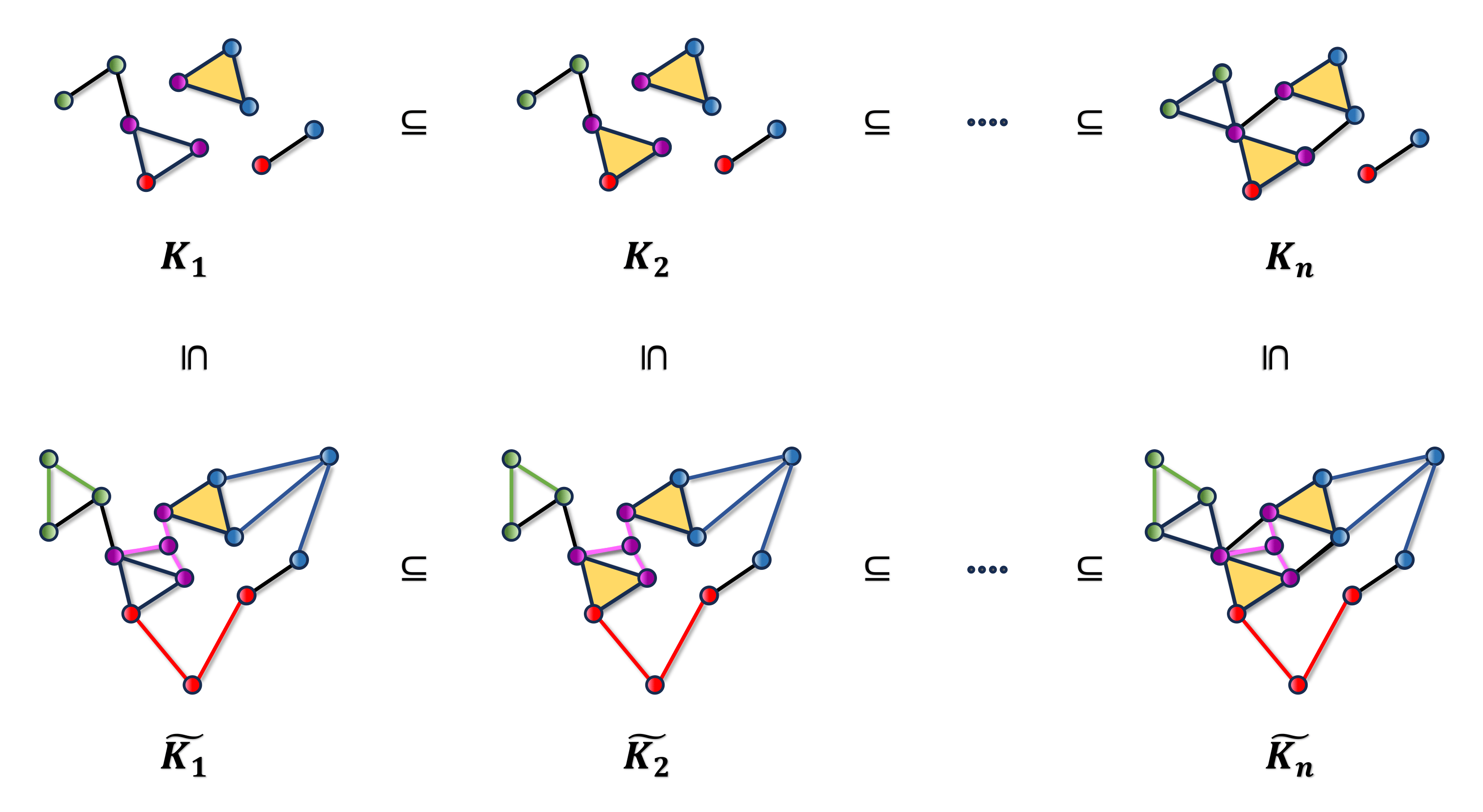}
	\caption{An illustrative example of a simplicial complex filtration $K_1 \subseteq K_2 \subseteq \cdots \subseteq K_n$ and the induced filtration $\widetilde{K_1} \subseteq \widetilde{K_2} \subseteq \cdots \subseteq \widetilde{K_n}$ that are homotopically equivalent to the filtration $\overline{K_1} \subseteq \overline{K_2} \subseteq \cdots \subseteq \overline{K_n}$. Three colored groups (blue, red, green, and purple) of vertices are considered equivalence classes of vertices within the filtration $K_1 \subseteq K_2 \subseteq \cdots \subseteq K_n$. Each $\widetilde{K_i} = K_i \cup S$ is obtained as the union of $K_i$ and the gluing stars $S$.}	
	\label{fig:two_filtrations}
\end{figure}

\begin{figure}
\centering
\centerline{\includegraphics[width=1\textwidth]{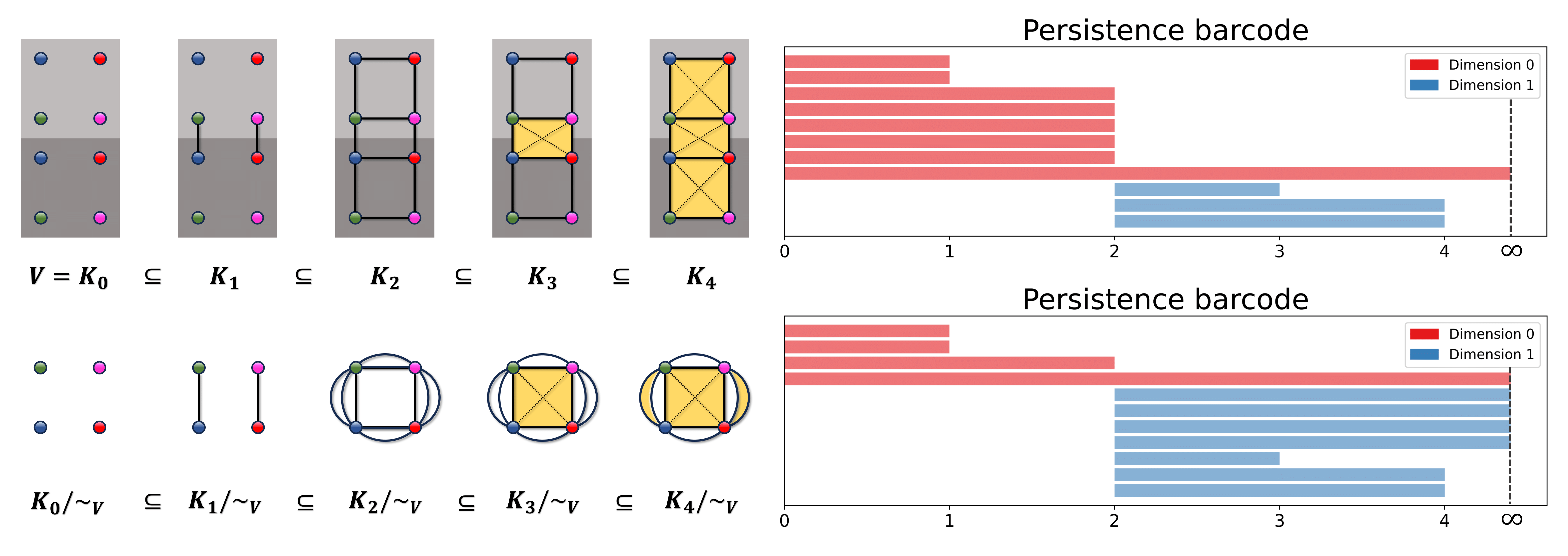}}
\caption{An illustration of a simplicial complex filtration, the induced QC filtration, and the corresponding persistence barcodes. Specifically, the simplicial complex $K_\bullet$ is a Vietoris--Rips simplicial complex built on a finite 2-periodic motif $V$ consisting of $8$ points. The periodic equivalence relation $\sim_V$ classifies the $8$ points into 4 classes annotated by different colors. Using 2D planar representations, the induced QC filtration $\overline{K_\bullet}$ is depicted in the second row. The corresponding persistence barcodes of $K_\bullet$ and $\overline{K_\bullet}$ are shown on the right-hand side of the figure. In particular, ${\rm PB}_0(K_\bullet) = \{ (0,1), (0,1), (0,2), (0,2), (0,2), (0,2), (0,2), (0,\infty) \}$, ${\rm PB}_1(K_\bullet) = \{ (2,3), (2,4), (2,4) \}$, ${\rm PB}_0(\overline{K_\bullet}) = \{ (0,1), (0,1), (0,2), (0,\infty) \}$, and ${\rm PB}_1(\overline{K_\bullet}) = \{ (2,3), (2,4), (2,4), (2, \infty), (2, \infty), (2, \infty), (2, \infty) \}$. These barcodes are generated using the Gudhi package~\cite{gudhi:urm}.}
\label{fig:quotient-complex ph}
\end{figure}

\paragraph{Persistent Homology of QC Filtration}
We introduce the construction of the QC filtration based on the following setting. In detail, let $\emptyset \subseteq K_1 \subseteq K_2 \subseteq \cdots \subseteq K_n \subseteq \mathbb{R}^d$ be a filtration of simplicial complexes, where $K_1$ is based on a discrete vertex set $V \subseteq \mathbb{R}^n$. Let $\sim_V$ be an equivalence relation $V$ with finitely many equivalence classes. Note that $V$ is a subcomplex of $X_i$ for each $i$, with an extended equivalence relation on each $X_i$. More precisely, as introduced in Section \ref{SI: Topological Background}, for $x, x' \in K_i$, $x \sim_V x'$ if, and only if either $x = x'$ or $x \sim_V x'$ with $x, x' \in V$. In particular, the canonical map $K_i/\sim_V \rightarrow K_{i+1}/\sim_V$ defines an inclusion. That is, 
\begin{equation*}
\emptyset \subseteq \overline{K_1} \hookrightarrow \overline{K_2} \hookrightarrow \overline{K_3} \hookrightarrow \cdots \hookrightarrow \overline{K_n}   
\end{equation*}
forms a filtration of quotient complexes. Consequently, by defining $\overline{K_i} = K_i/\sim_V$, the PH of the filtration $\overline{K_\bullet}$ is established. Incorporating the adjunction spaces $\widetilde{K_\bullet}$, the following commutative network is established.
\begin{equation}
\label{Eq. Commutative diagram of simplicial complexes and quotient complexes}
\xymatrix@+1.0em{
                & K_1
                \ar@{^{(}->}[r]^{}
                \ar@{^{(}->}[d]_{}
                & K_2
                \ar@{^{(}->}[r]^{}
                \ar@{^{(}->}[d]^{}
                & K_3
                \ar@{^{(}->}[r]^{}
                \ar@{^{(}->}[d]^{}
                & \cdots
                \ar@{^{(}->}[r]^{}
                & K_n
                \ar@{^{(}->}[d]^{}
                \\
                & \widetilde{K_1}
                \ar@{^{(}->}[r]^{}
                \ar[d]_{\widehat{r_1}}
        	& \widetilde{K_2}
				\ar@{^{(}->}[r]^{}
                \ar[d]_{\widehat{r_2}}
				& \widetilde{K_3}
                \ar@{^{(}->}[r]^{}
                \ar[d]_{\widehat{r_3}}
                & \cdots
                \ar@{^{(}->}[r]^{}
                & \widetilde{K_n}
                \ar[d]_{\widehat{r_n}}
                \\
                & \overline{K_1}
                \ar@{^{(}->}[r]^{}
        	& \overline{K_2}
				\ar@{^{(}->}[r]^{}
				& \overline{K_3}
                \ar@{^{(}->}[r]^{}
                & \cdots
                \ar@{^{(}->}[r]^{}
                & \overline{K_n}
				}
\end{equation}
It consists of topological spaces and continuous maps, encompassing inclusion maps and homotopy equivalences $\widehat{r_l}$ ($l = 1, 2, ..., n$). Especially, by the definition of $\widehat{r_l}$ depicted as in~\eqref{Eq. Induced map from K cupg S to K cupf Z}, the bottom rectangles commute. Furthermore, by utilizing the notation introduced in Theorem \ref{Theorem: onto, one-one, and iso}, the network in \eqref{Eq. Commutative diagram of simplicial complexes and quotient complexes} induces the following network of $q$-th homologies and the induced homomorphisms:
\begin{equation}
\label{Eq. Induced ladder tower}
\xymatrix@+1.0em{
                & H_q(K_1)
                \ar[r]^{}
                \ar[d]_{}
                & H_q(K_2)
                \ar[r]^{}
                \ar[d]^{}
                & H_q(K_3)
                \ar[r]^{}
                \ar[d]^{}
                & \cdots
                \ar[r]^{}
                & H_q(K_n)
                \ar[d]^{}
                \\
                & H_q(\widetilde{K_1})
                \ar[r]^{}
                \ar[d]_{H_q(\widehat{r_1})}
        	& H_q(\widetilde{K_2})
				\ar[r]^{}
                \ar[d]_{H_q(\widehat{r_2})}
				& H_q(\widetilde{K_3})
                \ar[r]^{}
                \ar[d]_{H_q(\widehat{r_3})}
                & \cdots
                \ar[r]^{}
                & H_q(\widetilde{K_n})
                \ar[d]_{H_q(\widehat{r_n})}
                \\
                & H_q(\overline{K_1})
                \ar[r]^{}
        	& H_q(\overline{K_2})
				\ar[r]^{}
				& H_q(\overline{K_3})
                \ar[r]^{}
                & \cdots
                \ar[r]^{}
                & H_q(\overline{K_n})
				}    
\end{equation}
Especially, maps $H_q(\widehat{r_l})$ are isomorphisms since $\widehat{r_l}$ are homotopy equivalences. In particular, $H_q(\widetilde{K_\bullet})$ and $H_q(\overline{K_\bullet})$ define the same persistence information, including the persistence barcodes. As a result, we shift our focus from the ladder $H_q(K_\bullet) \rightarrow H_q(\overline{K_\bullet})$ to the upper ladder $H_q(K_\bullet) \rightarrow H_q(\widetilde{K_\bullet})$ in~\eqref{Eq. Induced ladder tower}, and we analyze the barcode relations between flirtations $K_\bullet$ and $\widetilde{K_\bullet}$ in Section \ref{SI: Persistence Barcode Analysis}. Fig. \ref{fig:two_filtrations} exhibits a visible example of a filtration $K_1 \subseteq K_2 \subseteq \cdots \subseteq K_n$ and its associated filtration $\widetilde{K_1} \subseteq \widetilde{K_2} \subseteq \cdots \subseteq \widetilde{K_n}$.

\subsection{Persistence Barcode Analysis}
\label{SI: Persistence Barcode Analysis}
This section analyzes the homological relationships of the persistence barcodes of $H_q(K_\bullet)$ and $H_q(\overline{K_\bullet})$ in $q = 0$ and $1$. Specifically, the vertical maps of the top ladder in~\eqref{Eq. Induced ladder tower} are all isomorphisms for $q > 1$ (Theorem \ref{Theorem: onto, one-one, and iso}(c)), injective for $q = 1$ (Theorem \ref{Theorem: onto, one-one, and iso}(b)), and surjective for $q = 0$ (Theorem \ref{Theorem: onto, one-one, and iso}(a)). This observation allows us to focus our investigation on the persistence barcodes of filtrations $K_\bullet$ and $\widetilde{K_\bullet}$ in dimensions $0$ and $1$.

\paragraph{Analysis of the Barcode in Dimension 0}
We first investigate the persistence barcodes $H_0(K_\bullet)$ and $H_0(\widetilde{K_\bullet})$. In particular, let $\emptyset \subseteq K_1 \subseteq K_2 \subseteq \cdots \subseteq K_n \subseteq \mathbb{R}^d$, $V_1 \subseteq \mathbb{R}^d$ be a discrete vertex set of $K_1$, and $\sim_{V_1}$ be defined as in the previous sections. Then, we have the QC filtration $\emptyset \subseteq K_1/\sim_{V_1} \subseteq K_2/\sim_{V_1} \subseteq \cdots \subseteq K_n/\sim_{V_1}$ with $\widetilde{K_i} := K_i/\sim_{V_1}$ and the ladder
\begin{equation}
\label{Eq. 0-th ladder}
\xymatrix@+0.5em{
                & 0
                \ar[r]^{}
                & H_0(K_1)
                \ar[r]^{\phi_1}
                \ar@{->>}[d]_{\theta_{0,1}}
                & H_0(K_2)
                \ar[r]^{\phi_2}
                \ar@{->>}[d]_{\theta_{0,2}}
                & H_0(K_3)
                \ar[r]^{\phi_3}
                \ar@{->>}[d]_{\theta_{0,3}}
                & \cdots
                \ar[r]^{}
                & H_0(K_n)
                \ar@{->>}[d]_{\theta_{0,n}}
                \\
                & 0
                \ar[r]^{}
                & H_0(\widetilde{K_1})
                \ar[r]^{\psi_1}
        	& H_0(\widetilde{K_2})
				\ar[r]^{\psi_2}
				& H_0(\widetilde{K_3})
                \ar[r]^{\psi_3}
                & \cdots
                \ar[r]^{}
                & H_0(\widetilde{K_n})
                }    
\end{equation}
of the 0-th homology spaces and the $\mathbb{Z}_2$-linear transformations $\phi_\bullet$, $\psi_\bullet$, $\theta_{0,\bullet}$ induced by the inclusion maps.

Especially, if simplicial complexes $K_1, K_2, ..., K_n$ share the same vertex set $V = V_1$, then every persistence interval in the persistence barcode ${\rm PB}_0(X_\bullet)$ has the birth value $0$ since every connected component in each $K_i$ can be represented by a single vertex. However, to consider a more general situation, we slightly release this assumption $K_1, K_2, ..., K_n$ may have different vertex sets and allow intervals in $H_0(K_\bullet)$ and $H_0(\widetilde{K_\bullet})$ to have non-zero birth values.

We utilize the following notation to analyze the barcodes ${\rm PB}_0(K_\bullet)$ and ${\rm PB}_0(\widetilde{K_\bullet})$. We define $V_1 = V$ and set $V_l$ as the vertex set of $K_l$ for $l = 2, 3, ..., n$. In particular, $V_1 \subseteq V_2 \subseteq \cdots \subseteq V_n$. The partition of equivalence classes of $\sim_{V_1}$ in $V_1$ is denoted as $V_1 = \bigsqcup_{j = 1}^k V_{1,j}$.

\begin{theorem}
\label{Theorem: Pre-thm of 0-ladder}
Let $\emptyset \subseteq K_1 \subseteq K_2 \subseteq \cdots \subseteq K_n$ be a filtration of simplicial complexes embedded in $\mathbb{R}^d$, $V = V_1$ be a discrete set which serves as the vertex set of $K_1$, and $V = \bigsqcup_{j = 1}^k V_{1,j}$ the decomposition of equivalence classes defined by an equivalence relation $\sim_{V_1}$ on $V_1$. Then, for the ladder established in~\eqref{Eq. 0-th ladder}, the following properties follow:
\begin{itemize}
\item[\rm (a)] If $t$ has a birth value of $b$ in $H_0(\widetilde{K_\bullet})$, then there is an $s \in H_0(K_b)$ such that $t = \theta_{0,b}(s)$. In particular, $s$ is also born at $b$ within the filtration $H_0(K_\bullet)$.
\item[\rm (b)] If $t \in H_0(\widetilde{K_\bullet})$ has a persistence interval $(b, d)$ in $H_0(\widetilde{K_\bullet})$, then $t$ can be expressed as the equivalence class $\langle c \rangle$ in $H_0(\widetilde{K_b})$, where $c$ is a $0$-chain of vertices in $V_b \setminus V_{b-1}$. Especially, $c$ also represents an equivalence class $s = \langle c \rangle$ in $H_0(K_b)$, which also has a persistence interval $(b, d)$ in $H_0(K_\bullet)$.
\item[\rm (c)] If $t \in H_0(\widetilde{K_\bullet})$ has a persistence interval $(b, \infty)$ in $H_0(\widetilde{K_\bullet})$, then the element $s$ defined as in (b) has the same persistence interval $(b, \infty)$ in $H_0(K_\bullet)$.
\end{itemize}
\end{theorem}
\begin{proof}
(a) If $t$ has a birth value of $b$ in $H_0(\widetilde{K_\bullet})$, then $t \in H_0(\widetilde{K_b}) \setminus {\rm im}(\psi_{b-1})$. Because $\theta_{0,b}$ is onto, there is an $s \in H_0(K_b)$ such that $t = \theta_{0,b}(s)$. Suppose $s$ is not born at $b$ within $H_0(K_\bullet)$, then $s \in {\rm im}(\phi_{b-1})$. Then, $t \in {\rm im}(\theta_{0,b} \circ \phi_{b-1}) = {\rm im}(\psi_{b} \circ \theta_{0,b-1}) \subseteq  {\rm im}(\psi_{b})$. This shows that $t$ is not born at moment $b$, a contradiction.

(b) Because $\theta_{0,b}$ is onto, there is a $0$-chain $c$ such that $\theta_{0,b}(\langle c \rangle) = t$, where $s := \langle c \rangle$ denotes the equivalence class in $H_0(K_b)$ that is represented by $c$. By definition, $c$ can be represented as a formal sum of vertices $v_1 + v_2 + \cdots + v_m$ in $V_b$. By assertion (a), the element $s$ in $H_0(K_\bullet)$ must have a birth value of $b$, at least one the vertices in $\{ v_1, ..., v_m \}$ belongs to $V_{b} \setminus V_{b-1}$. If $\phi_{b,d-1}(s) \in {\rm im}(\phi_{b-1,d-1})$, then we have $\psi_{b,d-1}(t) = (\psi_{b,d-1} \circ \theta_{0,b})(s) = (\theta_{0,d-1} \circ \phi_{b,d-1})(s) \in {\rm im}(\theta_{0,d-1} \circ \phi_{b-1,d-1}) = {\rm im}(\psi_{b-1,d-1} \circ \theta_{0,b-1}) \subseteq {\rm im}(\psi_{b-1,d-1})$. This is impossible since $t$ doesn't die at $d-1$, and we deduce that $s$ has a death value $\geq d$.

On the other hand, by the death value of $t$, we have $\psi_{b,d}(t) = (\psi_{b-1, d} \circ \theta_{0,b-1})(u)$ for some $u \in H_{0}(K_{b-1})$, where $u$ is a formal sum $w_1 + w_2 + \cdots + w_r$ of vertices in $V_{b-1}$. In particular,
\begin{equation*}
v_1 + \cdots  + v_m = w_1 + w_2 + \cdots + w_r + \partial_1(\alpha) + \partial_1(\beta)  
\end{equation*}
for some $\alpha \in C_1(S)$ and $\beta \in C_1(K_b)$. Because the $0$-chains $v_1 + \cdots + v_m$, $w_1 + \cdots + w_r$, and $\partial_1(\beta)$ contain no vertices in $Z$. The $0$-chain $\partial_1(\alpha)$ contains no vertices in $Z$; in particular, $\partial_1(\alpha)$ is a $0$-chain consists of vertices in $V_1 \subseteq V_{b-1}$. This shows that
\begin{equation*}
\phi_{b,d}(s) = \langle v_1 + \cdots  + v_m \rangle = \langle w_1 + w_2 + \cdots + w_r + \partial_1(\alpha) \rangle \in {\rm im}(\phi_{b-1, d}).  
\end{equation*}
This shows that $s$ dies at the moment $d$ and we conclude that $s$ has a persistence interval $(b,d)$ in $H_0(K_\bullet)$.

(c) Suppose $t$ has a persistence interval $(b,\infty)$ and $s$ the element defined as in (b). By the proof of (b), if $s$ dies at some moment $d$, then $t$ must have a death value $\leq d$, which is impossible.
\end{proof}

\begin{corollary}
\label{Corollary: Pre-thm of 0-ladder}
Let $K_\bullet$ and $\widetilde{K_\bullet}$ be defined as in Theorem \ref{Theorem: Pre-thm of 0-ladder}. Let $H_0(K_\bullet)$ and $H_0(\widetilde{K_\bullet})$ be the $0$-th persistent homologies of the filtrations $K_\bullet$ and $\widetilde{K_\bullet}$ as shown in \ref{Eq. 0-th ladder}. Then, ${\rm PB}_0(\widetilde{K_\bullet})$ is a subset of ${\rm PB}_0(K_\bullet)$, i.e.,
\begin{equation*}
{\rm PB}_0(\widetilde{K_\bullet}) \subseteq {\rm PB}_0(K_\bullet).
\end{equation*}
\end{corollary}
\begin{proof}
By Theorem \ref{Theorem: Pre-thm of 0-ladder}-(b), every element which has a persistence interval $(b,d)$ in $H_0(\widetilde{K_\bullet})$ admits an element in $H_0(K_\bullet)$ that has the same persistence interval $(b,d)$. Moreover, suppose $\{ t_1, t_2, ..., t_m \}$ is a linear independent set in $H_0(\widetilde{K_b}) \setminus {\rm im}(\psi_{b-1})$, there are $s_1, s_2, ..., s_m \in H_0(K_b) \setminus {\rm im}(\phi_{b-1})$ such that $\theta_{0,b}(s_i) = t_i$ for each $i$. Because $\theta_{0,b}$ is linear, set $\{ s_1, s_2, ..., s_m \}$ is also linearly independent. In other words, each $t_i$ corresponds to a persistence interval in $H_0(K_\bullet)$, which represents $s_i$. In summary, ${\rm PB}_0(\widetilde{K_\bullet})$ is a subset of ${\rm PB}_0(K_\bullet)$.
\end{proof}

Corollary \ref{Corollary: Pre-thm of 0-ladder} elucidates that the barcode details of ${\rm PB}_0(\widetilde{K_\bullet})$ are intrinsically incorporated within the persistence barcode ${\rm PB}_0(K_\bullet)$. In particular, it encapsulates the merging relationships inherent in the original filtration, unaffected by the quotient relation defined on the vertex set $V$. This unveils a more fundamental topology that is independent of the quotient relation, offering a clearer depiction of the persistence within the original filtration.

\paragraph{Analysis of Barcode in Dimension 1} 
The previous part and observations establish the inclusion relationship ${\rm PB}_0(\widetilde{K_\bullet}) \subseteq {\rm PB}_0(K_\bullet)$. Similarly, an analogous version holds for the first persistence barcodes, i.e., ${\rm PB}_1(K_\bullet) \subseteq {\rm PB}_1(\widetilde{K_\bullet})$. We formalize this dual observation in the following theorem.

\begin{theorem}
\label{Theorem: barcode information in the included persistent homologies}
Let $\emptyset \subseteq K_1 \subseteq K_2 \subseteq \cdots \subseteq K_n$ be a filtration of simplicial complexes in $\mathbb{R}^d$.  Let $V$ be a discrete set, serving as the vertex set of $K$. Let $\sim_V$ be an equivalence relation on $V$ with finite equivalence classes $V_1, V_2, ..., V_k$. Then the following ladder
\begin{equation}
\label{Eq. H1 Ladder}
\xymatrix@+1.0em{
                & H_1(K_1)
                \ar[r]^{\phi_1}
                \ar@{^{(}->}[d]_{}
                & H_1(K_2)
                \ar[r]^{\phi_2}
                \ar@{^{(}->}[d]_{}
                & H_1(K_3)
                \ar[r]^{\phi_3}
                \ar@{^{(}->}[d]_{}
                & \cdots
                \ar[r]^{}
                & H_1(K_n)
                \ar@{^{(}->}[d]_{}
                \\
                & H_1(\widetilde{K_1})
                \ar[r]^{\psi_1}
        	& H_1(\widetilde{K_2})
				\ar[r]^{\psi_2}
				& H_1(\widetilde{K_3})
                \ar[r]^{\psi_3}
                & \cdots
                \ar[r]^{}
                & H_1(\widetilde{K_n})
				}    
\end{equation}
of the $1$-th persistent homology $H_1(K_\bullet)$ and $H_1(\widetilde{K_\bullet})$ are derived. Let $b, d \in \{ 1, 2, ..., n \}$. If $s \in H_1(K_b)$ has the persistence interval $(b,d)$ within $H_1(K_\bullet)$, then $s$ has the same birth and death moments $(b,d)$ within $H_1(\widetilde{K_\bullet})$. Furthermore, if $s$ has a death value $\infty$ in $H_1(K_\bullet)$, then $s$ also has a death value $\infty$ in $H_1(\widetilde{K_\bullet})$.    
\end{theorem}
\begin{proof}
All the homologies are vector spaces over the binary field $\mathbb{Z}_2$. Firstly, we prove that $s$ is born with the birth value of $b$ in $H_1(\widetilde{K_\bullet})$. By definition of homology, $s$ can be represented as the equivalence class of a $1$-cycle $\gamma$, which is a $\mathbb{Z}_2$-linear combination of $1$-simplices in $K_b$. If $s$ is not born at $b$, then there is a $t \in H_1(\widetilde{K_{b-1}})$ such that $\psi_{b-1}(t) = s$. Let $S$ be the simplicial complex as the union of gluing stars, making $\widetilde{K_l} = K_l \cup S$.

Because $s \in H_1(K_b) \setminus {\rm im}(\phi_{b-1})$, the representative of the equivalence class $t$ must have the form $\alpha + \beta$, where $\alpha$ is a non-zero $1$-chain consisting of $1$-simplices in $S$, $\beta$ is a $1$-chain of $1$-simplices in $K_{b-1} \subseteq K_b$, and $\alpha + \beta$ is a $1$-cycle in $\widetilde{K_{b-1}}$. Then $\alpha + \beta - \gamma$ is the boundary of a $2$-simplex in $\widetilde{K_b}$, which is also a $2$-simplex in $K_b$. It is impossible since $\alpha \neq 0$ contains $1$-simplices in $S$ as summands and $\beta - \gamma$ contains only terms of $1$-simplices in $K_b$.

Subsequently, we show that $s$ has the death value of $d$ in the second persistent homology $H_1(\widetilde{K_\bullet})$. Because $s$ dies at $d$, there is a $t \in H_{1}(K_{b-1}) \subseteq H_1(\widetilde{K_{b-1}})$ such that $\phi_{b-1,d}(t) = \phi_{b,d}(s)$. First, the commutativity of the ladder shows that $\psi_{b,d}(s) = \phi_{b,d}(s) = \phi_{b-1,d}(t) = \psi_{b-1,d}(t)$. On the other hand, suppose $\psi_{b,d'}(s) = \psi_{b-1,d'}(t')$ for some $t' \in H_1(\widetilde{K_{b-1}})$ and $d'$ with $b < d' < d$. By the same argument in the first part, any $1$-cycle representative of $t'$ must be a $1$-cycle in $K_{b-1}$ since $s \in H_1(K_b)$. However, it contradicts the fact that $s$ has a death value in $H_1(K_\bullet)$.

Lastly, suppose $s$ has an infinite death value in $H_1(K_\bullet)$ but a finite death value $d \leq n$ in $H_1(\widetilde{K_\bullet})$. By the proof of the second part, $s$ must have a death value $d' \leq d$ in $H_1(K_\bullet)$, leading to a contradiction.
\end{proof}

\begin{corollary}
\label{Coro: inclusion relation of the first PBs-prototype}
Let $\emptyset \subseteq K_1 \subseteq K_2 \subseteq \cdots \subseteq K_n \subseteq \mathbb{R}^d$ and $V = \bigsqcup_{j = 1}^k V_j \subseteq K_1$ be defined as in Theorem \ref{Theorem: barcode information in the included persistent homologies}. Let $s$ be an element in $H_1(K_b)$. Then $s \in {\rm im}(\psi_{b-1})$ if and only if $s \in {\rm im}(\phi_{b-1})$. In particular, the canonical linear transformation
\begin{equation*}
\frac{H_1(K_b)}{{\rm im}(\phi_{b-1})} \longrightarrow  \frac{H_1(\widetilde{K_b})}{{\rm im}(\psi_{b-1})}  
\end{equation*}
is one-to-one. Therefore, $H_1(K_b)/{\rm im}(\phi_{b-1})$ can be identified as a subspace of $H_1(\widetilde{K_b})/{\rm im}(\psi_{b-1})$. 
\end{corollary}
\begin{proof}
By the commutativity of rectangles in \eqref{Eq. H1 Ladder}, $s \in {\rm im}(\psi_{b-1})$ if $s \in {\rm im}(\phi_{b-1})$. If $s \notin {\rm im}(\phi_{b-1})$, then $s$ is born at $b$ in $H_1(K_\bullet)$. By Theorem \ref{Theorem: barcode information in the included persistent homologies}, $s$ is also born at $b$ in $H_1(\widetilde{K_\bullet})$. That is, $s \notin {\rm im}(\psi_{b-1})$.    
\end{proof}

In essence, the persistent homology $H_1(\widetilde{K_{\bullet}})$ preserves the original information from $H_1(K_{\bullet})$ and introduces additional loop information resulting from the gluing process. Specifically, ${\rm PB}_1(K_{\bullet})$ constitutes a subset of ${\rm PB}_1(\widetilde{K_{\bullet}})$, i.e., ${\rm PB}_1(K_{\bullet}) \subseteq {\rm PB}_1(\widetilde{K_{\bullet}})$. On the other hand, ${\rm PB}_q(K_{\bullet}) = {\rm PB}_q(\widetilde{K_{\bullet}})$ holds for $q > 1$. This observation is succinctly summarized in the following corollary.

\begin{corollary}
\label{Coro: inclusion relation of the first PBs}
Let $\emptyset \subseteq K_1 \subseteq K_2 \subseteq \cdots \subseteq K_n \subseteq \mathbb{R}^d$ and $V = \bigsqcup_{j = 1}^k V_j \subseteq K_1$ be defined as in Theorem \ref{Theorem: barcode information in the included persistent homologies}. Then ${\rm PB}_1(K_{\bullet}) \subseteq {\rm PB}_1(\widetilde{K_{\bullet}})$ and ${\rm PB}_q(K_{\bullet}) = {\rm PB}_q(\widetilde{K_{\bullet}})$ for $q > 1$.        
\end{corollary}

In a more intuitive sense, each persistence interval $(b, d)$ in ${\rm PB}_1(\widetilde{K_\bullet}) \setminus {\rm PB}_1(K_{\bullet})$ corresponds to a loop structure involving edges within the gluing star $S$. Specifically, the presence of such a loop in ${\rm PB}_1(\widetilde{K_b})$ will not be fulfilled in $K_{b'}$ for all $b' \geq b$. This is due to the absence of $2$-simplices in $K_{b'}$ that possess $1$-dimensional faces in $S$. The implications of this observation are captured in the following two theorems.

\begin{theorem}
\label{Theorem: for infinite B1}
Consider the ladder defined in equation \eqref{Eq. H1 Ladder} and $b \in \{ 1,2, ..., n\}$, then every $s \in H_1(\widetilde{K_b}) \setminus H_1(K_b)$ satisfies $\psi_{b,d}(s) \neq 0$ for every $d \geq b$.
\end{theorem}
\begin{proof}
Let $S$ be the simplicial complex defined as above, then $\widetilde{K_b} = K_b \cup S$ for each $b \in \{ 1,2, ..., n\}$. Because $s \in H_1(\widetilde{K_b}) \setminus H_1(K_b)$, it can be represented as a $1$-chain $\alpha + \beta$, where $\alpha \in C_1(S)$, $\beta \in C_1(K_b)$, and $\partial_1(\alpha + \beta) = 0$. Moreover, $\alpha \neq 0$ since $s \notin H_1(K_b)$. If $\psi_{b,d}(s) = 0$ for some $d \in \{ 1,2, ..., n\}$ with $d \geq b$, then there is a $2$-cycle $\delta \in C_2(\widetilde{K_d})$ such that $\partial_2(\delta) = \alpha + \beta$ in $C_1(\widetilde{K_d})$. Because $C_2(\widetilde{K_d}) = C_2(K_d)$, $\alpha + \beta \in C_1(K_d)$ with $\alpha \neq 0$. It is impossible since $K_d$ contains no $1$-simplices in $S$.    
\end{proof}

According to Corollary \ref{Coro: inclusion relation of the first PBs-prototype}, it follows that $H_1(K_b)/{\rm im}(\phi_{b-1}) \subseteq H_1(\widetilde{K_b})/{\rm im}(\psi_{b-1})$. Specifically, the persistence intervals with a birth value of $b$ in $H_1(\widetilde{K_\bullet})$ can be categorized into two groups: intervals corresponding to $H_1(K_b)$ and intervals represented by an element $s \in H_1(\widetilde{K_b})$ where $s \notin H_1(K_b) + {\rm im}(\psi_{b-1})$. This observation is captured in the following theorem.

\begin{theorem}
\label{Main Theorem of the infinite barcodes in the first QC PH}
Consider the ladder defined in equation \eqref{Eq. H1 Ladder}. Let $s \in H_1(\widetilde{K_b})$ be an element with $s \notin H_1(K_b) + {\rm im}(\psi_{b-1})$, then $s$ has a persistence interval $(b,\infty)$ in $H_1(\widetilde{K_\bullet})$.
\end{theorem}
\begin{proof}
Because $s \notin H_1(K_b) + {\rm im}(\psi_{b-1})$, $s$ is born at the moment $b$ in the persistent homology $H_1(\widetilde{K_\bullet})$. To investigate the death value of $s \in H_1(\widetilde{K_b})$, we consider the ladder
\begin{equation*}
%\label{Eq. H1 Ladder-v2}
\xymatrix@+1.0em{
                & H_1(K_{b-1})
                \ar[r]^{\phi_{b-1}}
                \ar@{^{(}->}[d]_{}
                & H_1(K_b)
                \ar[r]^{\phi_{b,d}}
                \ar@{^{(}->}[d]_{}
                & H_1(K_d)
                \ar@{^{(}->}[d]_{}
                \\
                & H_1(\widetilde{K_{b-1}})
                \ar[r]^{\psi_{b-1}}
        	& H_1(\widetilde{K_b})
				\ar[r]^{\psi_{b,d}}
			& H_1(\widetilde{K_d})
                }    
\end{equation*}
that compresses the composition maps $\phi_{d-1} \circ \cdots \circ \phi_{b}$ and $\psi_{d-1} \circ \cdots \circ \psi_{b}$ into maps $\phi_{b,d}$ and $\psi_{b,d}$ (c.f.,~\eqref{Eq. Vector space sequence}). If $s$ dies at moment $d$ in $H_1(\widetilde{K_\bullet})$, then there is an $s' \in H_1(\widetilde{K_{b-1}})$ such that $\psi_{b-1,d}(s') = \psi_{b,d}(s)$.  We may represent $s$ and $s'$ by the equivalence classes $\langle \alpha + \beta \rangle$ and $\langle \alpha' + \beta' \rangle$, where $\alpha, \alpha' \in C_1(S)$, $\beta \in C_1(K_b)$, and $\beta' \in C_1(K_{b-1}) \subseteq C_1(K_b)$. This shows that 
\begin{equation*}
\alpha + \alpha' = \beta + \beta' + \gamma,    
\end{equation*}
where $\gamma$ is the boundary of a $2$-chain in $\widetilde{K_d}$. Because $\beta + \beta' + \gamma$ contains no $1$-simplices in $S$, and all $1$-simplices in $\alpha$ and $\alpha'$ are in $S$, we must have $\alpha = \alpha'$ and $\beta + \beta' = \gamma$ (over $\mathbb{Z}_2$). In particular, $\partial_1(\beta + \beta') = \partial_1(\gamma) = 0$, and this shows that $\beta + \beta'$ defines an equivalence class $\langle \beta + \beta' \rangle$ in $H_1(K_b)$. Finally, we conclude that
\begin{equation*}
\begin{split}
s + \psi_{b-1}(s') &= \langle \alpha + \beta \rangle + \langle \alpha' + \beta' \rangle = \langle \alpha + \beta + \alpha' + \beta' \rangle = \langle \beta + \beta' \rangle.
\end{split}    
\end{equation*}
Because $\langle \beta + \beta' \rangle \in H_1(K_b)$, $s = \langle \beta + \beta' \rangle + \psi_{b-1}(s') \in H_1(K_b) + {\rm im}(\psi_{b-1})$. However, this contradicts the assumption of $s$, leading us to conclude that $s$ has an infinite death value in the persistent homology $H_1(\widetilde{K_\bullet})$.
\end{proof}

In the context of Vietoris--Rips complexes, considering real numbers $0 < \epsilon_1 < \epsilon_2 < \cdots < \epsilon_n$ as filtration levels and a given point-cloud $V \subseteq K$, we establish the Vietoris--Rips complex filtration $\emptyset \subseteq K_{\epsilon_1} \subseteq K_{\epsilon_2} \subseteq \cdots \subseteq K_{\epsilon_n}$, with the initial condition $V \subseteq K_{\epsilon_1}$. Conversely, by introducing an equivalence relation on $V$, we derive the induced quotient complex filtration $\emptyset \subseteq \overline{K_{\epsilon_1}} \subseteq \overline{K_{\epsilon_2}} \subseteq \cdots \subseteq \overline{K_{\epsilon_n}}$. 

Especially, the 1-th PB ${\rm PB}_1(\overline{K_{\epsilon_\bullet}})$ is partitioned into two sub-multisets: ${\rm PB}_1^{\rm finite}(\overline{K_{\epsilon_\bullet}})$ and ${\rm PB}_1^{\infty}(\overline{K_{\epsilon_\bullet}})$, representing the collections of finite and infinite persistence intervals within ${\rm PB}_1(\overline{K_{\epsilon_\bullet}})$, respectively. By applying Corollary \ref{Coro: inclusion relation of the first PBs} and Theorem \ref{Main Theorem of the infinite barcodes in the first QC PH}, we deduce that ${\rm PB}_1^{\rm finite}(\overline{K_{\epsilon_\bullet}})$ is a subset of ${\rm PB}_1(K_{\epsilon_\bullet})$. In particular, the following identification holds: 
\begin{equation}
\label{Eq. PB1 decomposition}
{\rm PB}_1(\overline{K_{\epsilon_\bullet}}) = {\rm PB}_1^{\rm finite}(\overline{K_{\epsilon_\bullet}}) \cup {\rm PB}_1^{\infty}(\overline{K_{\epsilon_\bullet}}) = {\rm PB}_1(K_{\epsilon_\bullet}) \cup {\rm PB}_1^{\infty}(\overline{K_{\epsilon_\bullet}}).    
\end{equation}
Furthermore, if $\epsilon_n$ is large enough, then the PH $H_1(K_{\epsilon_\bullet})$ has no intervals of lifespan $\infty$. In particular, there is a one-to-one correspondence between ${\rm PB}_1(K_{\epsilon_\bullet})$ and ${\rm PB}_1^{\rm finite}(\overline{K_{\epsilon_\bullet}})$, i.e., ${\rm PB}_1(K_{\epsilon_\bullet}) = {\rm PB}_1^{\rm finite}(\overline{K_{\epsilon_\bullet}})$. 

Fig. \ref{fig:quotient-complex ph} illustrates an example of PBs for a simplicial complex filtration $K_\bullet$ and the induced QC filtration $\overline{K_\bullet}$. In this example, the filtration $K_0 \subseteq K_1 \subseteq K_2 \subseteq K_3 \subseteq K_4$ consists of 5 simplicial complexes within a $1 \times 2$ supercell in $\mathbb{R}^2$. Mathematically, the 0-th PB of $K_\bullet$ can be represented by
\begin{equation*}
{\rm PB}_0(K_\bullet) = \{ (0,1), (0,1), (0,2), (0,2), (0,2), (0,2), (0,2), (0,\infty) \},
\end{equation*}
which records the merging behavior of connected components within the filtration. Notably, 8 connected components represented by the vertices are born at moment 0, while two of them merge with other components at 1, and finally, all components merge together at moment 2. On the other hand, similar birth and death information can be utilized to capture the lifespan details of loop structures within $K_\bullet$. Specifically,
\begin{equation*}
{\rm PB}_1(K_\bullet) = \{ (2,3), (2,4), (2,4) \}.
\end{equation*}
It means that three loops are born at moment 2, one of which is eliminated at 3 by being filled with $2$-simplices, and two of which die at moment 4.

Based on the equivalence relation $\sim_V$ on the vertex set $V$, the second row of Fig. \ref{fig:quotient-complex ph} illustrates a 2D planar representation of the filtration $\overline{K_0} \subseteq \overline{K_1} \subseteq \overline{K_2} \subseteq \overline{K_3} \subseteq \overline{K_4}$. Especially,
\begin{equation*}
{\rm PB}_0(\overline{K_\bullet}) = \{ (0,1), (0,1), (0,2), (0,\infty) \}.
\end{equation*}
We note that the number of intervals in ${\rm PB}_0(\overline{K_\bullet})$ is half that of ${\rm PB}_0(K_\bullet)$ since $\sim_V$ separates all 8 points in $V$ into 4 equivalence classes based on periodicity. For the $1$-th PB of $\overline{K_\bullet}$, the quotient operation may generate additional loops. For instance, by merging the periodic vertices of $K_2$, $4$ additional loop structures are created, all with a birth value of 2 and an infinite death value. In summary,
\begin{equation*}
{\rm PB}_1(\overline{K_\bullet}) = \{ (2,3), (2,4), (2,4), (2, \infty), (2, \infty), (2, \infty), (2, \infty) \}.
\end{equation*}

In this study, we introduce a framework for constructing the QC filtration from a simplicial complex filtration. This involves applying quotient relations to the vertex set of the initial simplicial complex. Additionally, we investigate the relationships between the QC-based PB and the traditional PB. Our approach provides a detailed classification of persistence intervals, distinguishing between intervals originating within the original filtration and those generated by the quotient operation. 

While our primary focus in this study is on generating quotient complexes through the quotient relation applied to $0$-simplices, it is worth noting that exploring more complex quotient relations on higher-dimensional simplices, such as edges, faces, or tetrahedrons, may lead to more generalized theorems, as discussed in the sections. The investigation of these intricate quotient complexes represents a significant and promising direction for future research.

\section{Quotient Complex Descriptors (QCDs)}
\label{Appendix: Quotient Complex Descriptors (QCDs)}

\subsection{QC Representation in Finite Context}
In terms of computational efficiency, this study adopts a methodology to construct the finite periodic motif, QC filtration, and QC-based PH in alignment with the CIF data format for material structures. Regarding the finite periodic motif, our approach diverges from supercell extensions, such as the $3 \times 3 \times 3$ or $5 \times 5 \times 5$ extensions. Instead, our focus lies on the union of $M$ and the other three translated motifs.
\begin{equation}
\label{Eq. 4M-2}
\begin{split}
V &= M \cup (M + v_1) \cup (M + v_2) \cup (M + v_3).
\end{split}
\end{equation}
Here, $\mathcal{B} = \{ v_1, v_2, v_3 \}$ denotes the $\mathbb{R}$-basis specified in the material's CIF file. To succinctly and effectively represent the material, we predominantly collect points within the original motif $M$ and account for the material's periodicity by incorporating its union with the translated copies $M + v_i$. On the other hand, the bipartite distance $d: V \times V \rightarrow \mathbb{R} \cup \{ +\infty \}$ to characterize the interaction relationships between unit cells is defined as:
\begin{equation}
\label{Eq. Distance function of extended cells-v2}
d(u,v) = \begin{cases}
 	 +\infty &\quad\null\text{ if } u \notin M \text{ and 
 } v \notin M, \\
 	\Vert u - v \Vert_2 &\quad\null\text{ otherwise.}
 	\end{cases}   
\end{equation}
The Vietoris--Rips filtration $K_{\epsilon_\bullet}$ can be constructed by this distance function. In addition, by defining $v \sim_V w$ as $v - w \in \Lambda(\mathcal{B})$, the QC $\overline{K_{\epsilon_\bullet}}$ with $\overline{K_{\epsilon_i}} = K_{\epsilon_i}/\sim_V$ is defined.  For the computation of the QC filtration's PH based on Vietoris--Rips filtration, the primary parameter to determine is the maximum filtration level of radii. In terms of computation, we calculate the PH over the interval of $[0, 10]$ (\r{A}) to explore the linkages between the original and shifted neighboring cells.

The subsequent theorem demonstrates that, given the configuration of the periodic motif $V$ (\eqref{Eq. 4M}) and the distance function (\eqref{Eq. Distance function of extended cells-v2}), the periodic information can be encapsulated within the infinite persistence intervals of the 1-th PB in the proposed QC filtration. In other words, the QC-based persistent barcodes not only extend the conventional periodic information of the periodic motif but also provide a more intricate representation of periodic relationships within the crystal structure.

\begin{theorem}
\label{Theorem: Unit cell info as PB intervals}
Let $\mathcal{B} = \{ v_1, v_2, v_3 \}$ denote an $\mathbb{R}$-basis in $\mathbb{R}^3$ and $U$ be the unit cell spanned by $\mathcal{B}$. Let $M = \{ (0,0,0) \}$ be the motif set within $U$. Based on the point cloud $V$ in~\eqref{Eq. 4M}, bipartite distance in~\eqref{Eq. Distance function of extended cells-v1}, Vietoris--Rips complex $K_{\epsilon_\bullet}$, and the associated QC filtration $\overline{K_{\epsilon_\bullet}}$, there are persistence intervals $(b_1, \infty), (b_2, \infty)$, and $(b_3, \infty)$ in ${\rm PB}_1^\infty(\overline{K_{\epsilon_\bullet}})$ with $b_i = |v_i|$ for $i = 1,2,3$.    
\end{theorem}
\begin{proof}
For computational convenience, we identify $\overline{K_\epsilon}$ as the simplicial complex $\widetilde{K_\epsilon}$. By equation~\eqref{Eq. 4M-2}, the extended point cloud $V$ is the set $\{ v_0, v_1, v_2, v_3 \}$, where $v_0 = (0,0,0)$ is the origin in $\mathbb{R}^3$. In particular, we have $|v_i| = |v_i - v_0|$ for $i = 1,2,3$ and $v_{i} \sim v_j$ for all $i, j \in \{ 0,1,2,3 \}$. Moreover, we have $d(v_i, v_j) = \infty$ for $i \neq j$ in $\{ 1, 2, 3\}$. Since any three points in $V$ must include two points defining an infinite distance, $K_\epsilon$ and $\widetilde{K_\epsilon}$ contain no $2$-simplices for every $\epsilon > 0$. In particular, any nonzero $1$-cycle in $\widetilde{K_\epsilon}$ determines a nonzero element in $H_1(\widetilde{K_\epsilon};\mathbb{Z}_2)$, i.e., a loop structure in $\widetilde{K_\epsilon}$. Assume $|v_1| \leq |v_2| \leq |v_3|$ and set $\epsilon_i = |v_i|$ for $i = 1,2,3$, then the corresponding simplicial complexes $\widetilde{K_{r}}$ is
\begin{equation*}
\begin{split}
\widetilde{K_{r}} &= \left( \bigcup_{i=1}^3 \ [v_0,v_i] \right) \cup \left( \bigcup_{i = 0}^3 \ [z,v_i] \right).    
\end{split}
\end{equation*}
for $r \geq \epsilon_3$. In particular, the $1$-cycles $[z, v_i] + [z,v_0] + [v_0,v_i]$ with $i = 1,2,3$ are linearly independent in $C_1(\widetilde{K_r};\mathbb{Z}_2)$ since there are linearly independent as $1$-chains in $C_1(\widetilde{K_r};\mathbb{Z}_2)$ and $C_2(\widetilde{K_r};\mathbb{Z}_2) = 0$. Because $H_1(K_{\epsilon};\mathbb{Z}_2) = 0$ for every $\epsilon > 0$, the theorem immediately follows from Theorem \ref{Main Theorem of the infinite barcodes in the first QC PH}.    
\end{proof}

\subsection{Quotient Complex Descriptors (QCDs)}
We extract features for the persistence barcodes ${\rm PB}_0(\overline{K_\bullet})$, ${\rm PB}_1^{\rm finite}(\overline{K_\bullet})$, ${\rm PB}_1^{\infty}(\overline{K_\bullet})$, and ${\rm PB}_2(\overline{K_\bullet})$ using both statistical measures and Betti curves~\cite{umeda2017time}. Due to the distinct types of persistence intervals present in ${\rm PB}_0(\overline{K_\bullet})$, ${\rm PB}_1^{\rm finite}(\overline{K_\bullet})$, ${\rm PB}_1^{\infty}(\overline{K_\bullet})$, and ${\rm PB}_2(\overline{K_\bullet})$, the feature extraction process for these barcodes exhibits slight variations. In essence, our analysis delves into various aspects of birth, death, and lifespan information within these persistence barcodes. For the implementation, refer to the technique document released on GitHub (\url{https://github.com/peterbillhu/QCPH}). 

\paragraph{Statistical Measurements}
Given a collection $\mathcal{L}$ of real numbers, we consider the following statistical estimators as a quantitative summarization of $\mathcal{L}$: the maximum ($M$), minimum ($m$), $(\frac{1}{4}, \frac{2}{4}, \frac{3}{4})$-quartiles ($q_{\frac{1}{4}}, q_{\frac{2}{4}}, q_{\frac{3}{4}}$), average value ($\mu$), and standard deviation ($\sigma$). We call $(M, m, q_{\frac{1}{4}}, q_{\frac{2}{4}}, q_{\frac{3}{4}}, \mu, \sigma)$ the \textit{statistical descriptors} of the collection $\mathcal{L}$. Furthermore, for a collection $\mathcal{L}$ consisting of non-negative real numbers, we define the normalized $\mathcal{L}$ as the collection:
\begin{equation*}
\widetilde{\mathcal{L}} = \left \{ \frac{l}{\sum_{l' \in \mathcal{L}} l'} \ \bigg| \ l \in \mathcal{L} \right \}.
\end{equation*}
Since ${\rm PB}_0(\overline{K_\bullet})$, ${\rm PB}_1^{\rm finite}(\overline{K_\bullet})$, ${\rm PB}_1^{\infty}(\overline{K_\bullet})$, and ${\rm PB}_2(\overline{K_\bullet})$ encompass distinct types of persistence intervals, we adopt different sets $\mathcal{L}$ to extract features. Particularly, for ${\rm PB}_0(\overline{K_\bullet})$, we examine the statistical descriptors of the following collections:
\begin{itemize}
\item the collection of death numbers: $\mathcal{L}_d = \{ d \ | \ (0,d) \in {\rm PB}_0(\overline{K_\bullet}) \}$;
\item the collection of normalized death numbers: $\widetilde{\mathcal{L}_d}$. 
\end{itemize}
We emphasize that this definition aligns with the interval distribution of ${\rm PB}_0(\overline{K_\bullet})$, given that all persistence intervals in ${\rm PB}_0(\overline{K_\bullet})$ share a birth value of $0$. Conversely, for ${\rm PB}_1^{\infty}(\overline{K_\bullet})$, where all intervals have an infinite lifespan, the statistical descriptors examine the following sets:
\begin{itemize}
\item the collection of birth numbers: $\mathcal{L}_b = \{ b \ | \ (b,\infty) \in {\rm PB}_1^{\infty}(\overline{K_\bullet}) \}$;
\item the collection of normalized birth numbers: $\widetilde{\mathcal{L}_b}$. 
\end{itemize}
For the resistance barcodes $\mathcal{P} = {\rm PB}_1^{\rm finite}(\overline{K_\bullet})$ or ${\rm PB}_2(\overline{K_\bullet})$, all intervals are finite with various birth values. The statistical descriptors of $\mathcal{P}$ track the following birth-death information:
\begin{itemize}
\item the collection of birth numbers: $\mathcal{L}_b(\mathcal{P}) = \{ b \ | \ (b,d) \in \mathcal{P} \}$;
\item the collection of normalized birth numbers: $\widetilde{\mathcal{L}_b(\mathcal{P})}$;
\item the collection of death numbers: $\mathcal{L}_d(\mathcal{P}) = \{ d \ | \ (b,d) \in \mathcal{P} \}$;
\item the collection of normalized death numbers: $\widetilde{\mathcal{L}_d(\mathcal{P})}$;
\item the collection of middle points: $\mathcal{L}_{\frac{b+d}{2}}(\mathcal{P}) = \left\{ \frac{b+d}{2} \ | \ (b,d) \in \mathcal{P} \right\}$;
\item the collection of normalized middle points: $\widetilde{\mathcal{L}_{\frac{b+d}{2}}(\mathcal{P})}$;
\item the collection of lifespans: $\mathcal{L}_{d-b}(\mathcal{P}) = \{ d-b \ | \ (b,d) \in \mathcal{P} \}$;
\item the collection of normalized lifespans: $\widetilde{\mathcal{L}_{d-b}(\mathcal{P})}$;
\end{itemize}

In summary, the statistical descriptors $(M, m, q_{\frac{1}{4}}, q_{\frac{2}{4}}, q_{\frac{3}{4}}, \mu, \sigma)$ comprehensively capture the statistical indices of various collections of birth-death information within the persistence barcodes ${\rm PB}_0(\overline{K_\bullet})$, ${\rm PB}_1^{\rm finite}(\overline{K_\bullet})$, ${\rm PB}_1^{\infty}(\overline{K_\bullet})$, and ${\rm PB}_2(\overline{K_\bullet})$, constituting a fundamental part of the QCDs.

\paragraph{Betti Curve Information}
Given a persistence barcode $\mathcal{P}$ comprising persistence intervals $(b,d)$ range from the interval $[0,T]$ of filtration levels, we utilize two curves to encapsulate the information within $\mathcal{P}$: the \textit{Betti curve} and the \textit{normalized Betti curve}, denoted as ${\rm BC}(\mathcal{P})$ and ${\rm NBC}(\mathcal{P})$, respectively. Formally, they are defined as functions from $[0,T]$ to $\mathbb{R}$ with
\begin{equation*}
{\rm BC}(\mathcal{P})(t) = \# \{ (b,d) \in \mathcal{P} \ | \ b \leq t < d \}    \text{ and } {\rm NBC}(\mathcal{P})(t) = \frac{{\rm BC}(\mathcal{P})(t)}{\# \mathcal{P}}.
\end{equation*}
We consider the Betti curves and normalized Betti curves of the persistence barcodes ${\rm PB}_0(\overline{K_\bullet})$, ${\rm PB}_1^{\rm finite}(\overline{K_\bullet})$, ${\rm PB}_1^{\infty}(\overline{K_\bullet})$, ${\rm PB}_2(\overline{K_\bullet})$ as QC-based descriptors of the material. In application, the maximal filtration level $T$ is set to be $10$ (\r{A}) for computation efficiency.
 
\paragraph{Unit Cell Information}
Additionally, within our specialized QCDs, as illustrated in Theorem \ref{Theorem: Unit cell info as PB intervals}, we further enrich the feature set by considering $(|v_1|, |v_2|, |v_3|, |v_1 + v_2|, |v_1 + v_3|, |v_2 + v_3|, |v_1 + v_2 + v_3|)$ with $|v_1| \leq |v_2| \leq |v_3|$ to accentuate the periodic information within the material. This feature encapsulates edge length information of the unit cell representation and has proven to be a pivotal feature in GNN models designed for crystalline structures~\cite{yan2022periodic}.

\subsection*{Element-specific QCDs}

Capitalizing on the advantage of element-specific feature extraction from specific atomic systems or their combinations, particularly in the realm of bioinformatics molecules and perovskite data analysis~\cite{cang:2017integration, szocinski2021awegnn, mayr2021global, anand2022topological}, our approach meticulously captures QCDs of specific atomic systems and their atom combinations within a 2D perovskite. Specifically, we leverage the QCDs, which are derived from diverse atomic site combinations and types (i.e., atom sets), including the B-site ($\rm B$), X-site ($\rm X$), and their combinations (e.g., ${\rm A}_{\rm C}{\rm B}{\rm X}$). On the other hand, our consideration extends to QCDs affiliated with distinct atom types, spanning from $\rm O$ and $\rm N$ (associated with the A-site atom type) to $\rm Bi$, $\rm Cd$, $\rm Pb$, and $\rm Sn$, etc. (pertaining to the B-site atom type), and finally $\rm Cl$, $\rm Br$, and $\rm I$ (linked to the X-site atom type). The specifics of site combinations and atom types are elucidated in Table \ref{Table: Involving sites}.

 \begin{table}
\centering
\caption{This table provides comprehensive descriptions of the involved site combinations and atomic systems of element-specific QCDs. A total of 17 site combinations and atomic systems are considered in this study, based on the ${\rm A}{\rm B}{\rm X}$ structure found in 2D perovskites. The primary inorganic atoms featured in materials with DFT-band gaps include $\rm Bi$, $\rm Cd$, $\rm Ge$, $\rm Pb$, and $\rm Sn$. Notably, some materials encompass additional inorganic atoms beyond this set (e.g. $\rm Cu$), categorized under the $\rm B$ type. For a more detailed exploration of the material structures, please refer to the NMSE database~\cite{marchenko2020database}.}
\begin{tabular}{|lcc|}
 \hline
 Atom set & Category & Involving atom coordinates \\ % [0.5ex]
 \hline\hline
 ${\rm A}_{\rm C}{\rm B}$ & Site combinations & $\rm C$ atoms and inorganic atoms (e.g., ${\rm Bi}$, ${\rm Cd}$, ${\rm Cu}$, etc.)\\
 \hline
 ${\rm A}_{\rm C}{\rm X}$ & Site combinations & $\rm C$ atoms and halide atoms (${\rm Cl}$, ${\rm Br}$, ${\rm I}$)\\
 \hline
 ${\rm B}{\rm X}$ & Site combinations & inorganic atoms and halide atoms (${\rm Cl}$, ${\rm Br}$, ${\rm I}$)\\
 \hline
 ${\rm A}_{\rm C}{\rm B}{\rm X}$ & Site combinations & $\rm C$ atoms, inorganic atoms, and halide atoms (${\rm Cl}$, ${\rm Br}$, ${\rm I}$)\\
 \hline
 ${\rm C}$ & ${\rm A}$ & $\rm C$ atoms\\
 \hline
 ${\rm O}$ & ${\rm A}$ & $\rm O$ atoms\\
 \hline
 ${\rm N}$ & ${\rm A}$ & $\rm N$ atoms\\
 \hline
 ${\rm B}$ & ${\rm B}$ & Inorganic atoms (e.g., ${\rm Bi}$, ${\rm Cd}$, ${\rm Cu}$, ..., ${\rm Pb}$, ${\rm Sn}$, etc.)\\
 \hline
 ${\rm Bi}$ & ${\rm B}$ & ${\rm Bi}$ atoms\\
 \hline
 ${\rm Cd}$ & ${\rm B}$ & ${\rm Cd}$ atoms\\
 \hline
 ${\rm Ge}$ & ${\rm B}$ & ${\rm Ge}$ atoms\\
 \hline
 ${\rm Pb}$ & ${\rm B}$ & ${\rm Pb}$ atoms\\
 \hline
 ${\rm Sn}$ & ${\rm B}$ & ${\rm Sn}$ atoms\\
 \hline
 ${\rm X}$ & ${\rm X}$ & Halide atoms (${\rm Cl}$, ${\rm Br}$, ${\rm I}$)\\
 \hline
 ${\rm Cl}$ & ${\rm X}$ & ${\rm Cl}$ atoms\\
 \hline
 ${\rm Br}$ & ${\rm X}$ & ${\rm Br}$ atoms\\
 \hline
 ${\rm I}$ & ${\rm X}$ & ${\rm I}$ atoms\\
 \hline
\end{tabular}
\label{Table: Involving sites}
\end{table}
\subsection*{QCD-based ML Model}
In the realm of data-driven material AI models, a two-fold classification emerges, encompassing featurization-based machine learning models and end-to-end deep learning models. Featurization, also known as feature engineering, involves the creation of a spectrum of structural, physical, chemical, and biological attributes that encapsulate intrinsic material information~\cite{schutt2014represent,ramprasad2017machine,isayev2015materials,huan2015accelerated,damewood2023representations}. These attributes, termed molecular descriptors or fingerprints, encompass conventional material descriptors like atomic, ionic, and orbital radii, as well as factors such as tolerance and octahedral factors, packing factor, crystal structure measurements, and surface/volume features. Additionally, they encompass physical properties including ionization potential, ionic polarizability, electron affinity, Pauling electronegativity, valence orbital radii, HOMO, and LUMO. A notable instance is the molecular fingerprint, an extensive vector composed of systematically generated features primarily derived from molecular structures. Prominent material molecular fingerprints include the Coulomb matrix~\cite{himanen2020dscribe}, Ewald sum matrices~\cite{rupp2012fast}, many-body tensor representation (MBTR)~\cite{huo2022unified}, smooth overlap of atomic positions (SOAP)~\cite{bartok2013representing}, and atom-centered symmetry functions (ACSF)~\cite{behler2011atom}. These descriptors, serving as input features for machine learning models, find versatile applications in predicting diverse material properties, with a specific focus on the prediction of the bandgap of 2D perovskites~\cite{marchenko2020database, mayr2021global}.

Geometric deep learning models offer a feature-free approach and a seamless end-to-end method for directly understanding material properties. Geometric deep learning, introduced to handle non-Euclidean data like graphs, networks, and manifolds, achieves this by integrating geometric and topological insights into deep learning frameworks \cite{atz2021geometric,bronstein2017geometric,bronstein2021geometric,masci2015geodesic}. Among these models, graph neural networks (GNNs) stand out, finding extensive application in node classification, link prediction, graph classification, and property forecasting \cite{hamilton2017inductive,vashishth2019composition,velivckovic2017graph,welling2016semi}. Recent advancements have produced diverse GNN models for analyzing material properties, examples being SchNet \cite{schmidt2017predicting}, crystal graph convolutional neural networks (CGCNN) \cite{xie2018crystal}, materials graph network (MEGNet) \cite{chen2019graph}, improved crystal graph convolutional neural networks (iCGCNN) \cite{park2020developing}, global attention-based graph neural network (GATGNN) \cite{louis2020graph}, crystal graph attention network (CYATT) \cite{schmidt2021crystal}, compositionally restricted attention-based network (CrabNet) \cite{wang2021compositionally}, neural equivariant interatomic potentials (NequIP) \cite{batzner20223}, atomistic line graph neural network (ALIGNN) \cite{choudhary2021atomistic}, Matformer \cite{yan2022periodic}, MatDeepLearn \cite{fung2021benchmarking}, and SIGNNA \cite{na2023substructure}. In general, GNN models for materials analysis can be categorized into two primary groups: potential-energy-based GNNs, encompassing SchNet \cite{schmidt2017predicting}, MEGNet \cite{chen2019graph}, and NequIP \cite{batzner20223}, which employ tailored architectures to approximate potential energy for predicting material properties, and chemical-composition-based GNNs, including CYATT \cite{schmidt2021crystal} and CrabNet \cite{wang2021compositionally}, which directly construct deep learning architectures using chemical formulas or structures as inputs.

In light of the constrained availability of public 2D perovskite datasets, we opt to utilize a featurization-based machine learning model as our primary downstream approach. This decision is prompted by the limited accessibility of 2D perovskite datasets. Furthermore, featurization-based machine learning models offer significant advantages for testing and incorporating the proposed topological descriptors, particularly when compared to end-to-end models. Their appropriateness for integrating and assessing the effectiveness of the proposed topological descriptors stands out as a notable advantage in our methodology.

\section{2D Perovskites and Density Functional Calculation}

\begin{figure}
\centering
\centerline{\includegraphics[width=1\textwidth]{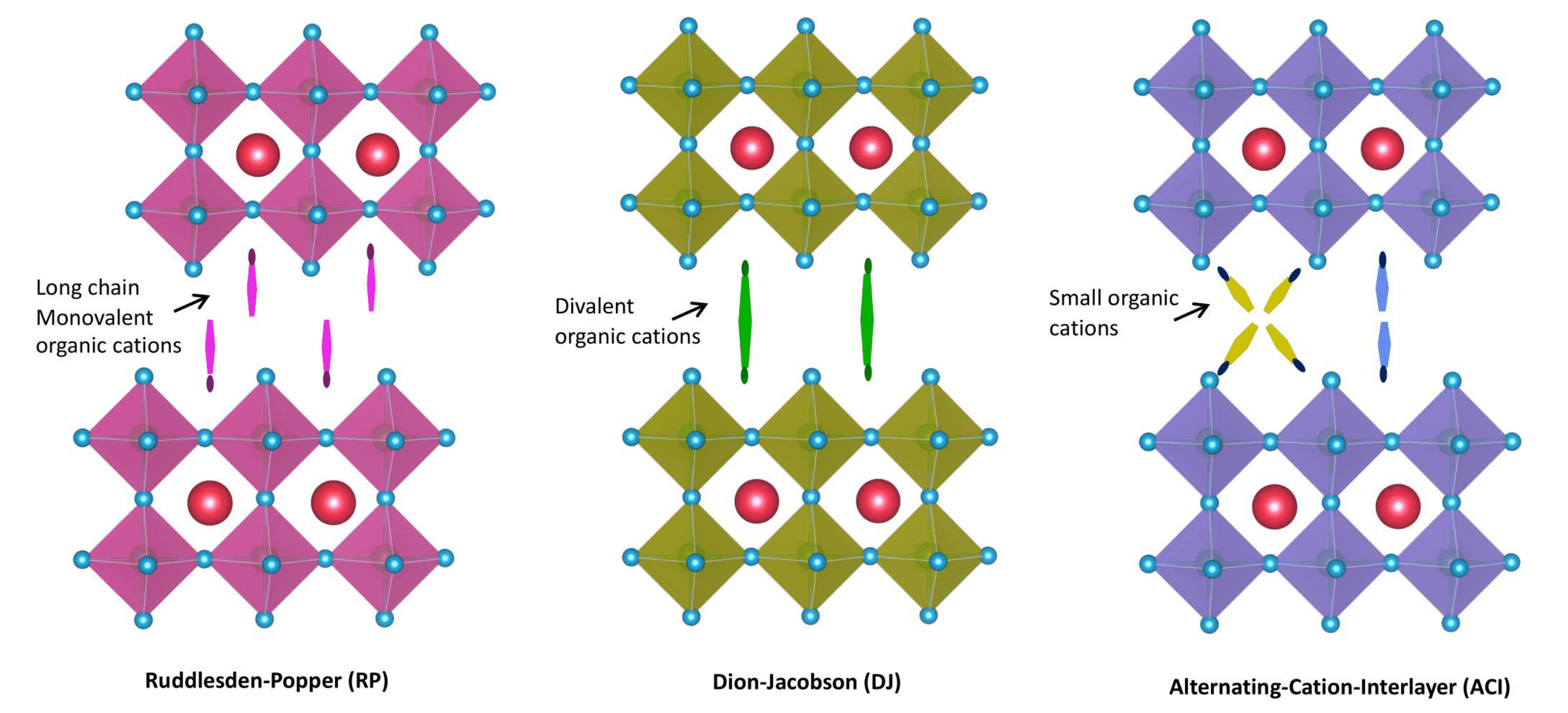}}
\caption{Illustration of three phases of 2D perovskites, the Ruddlesden-Popper (RP), the Dion-Jacobson (DJ), and the alternating-cation-interlayer (ACI) 2D Perovskites. Visualization of the crystal structure is produced by the VESTA software~\cite{momma2011vesta}.}
\label{fig:Phases_of_2D_Perovskites}
\end{figure}

\begin{figure}
\centering
\centerline{\includegraphics[width=1\textwidth]{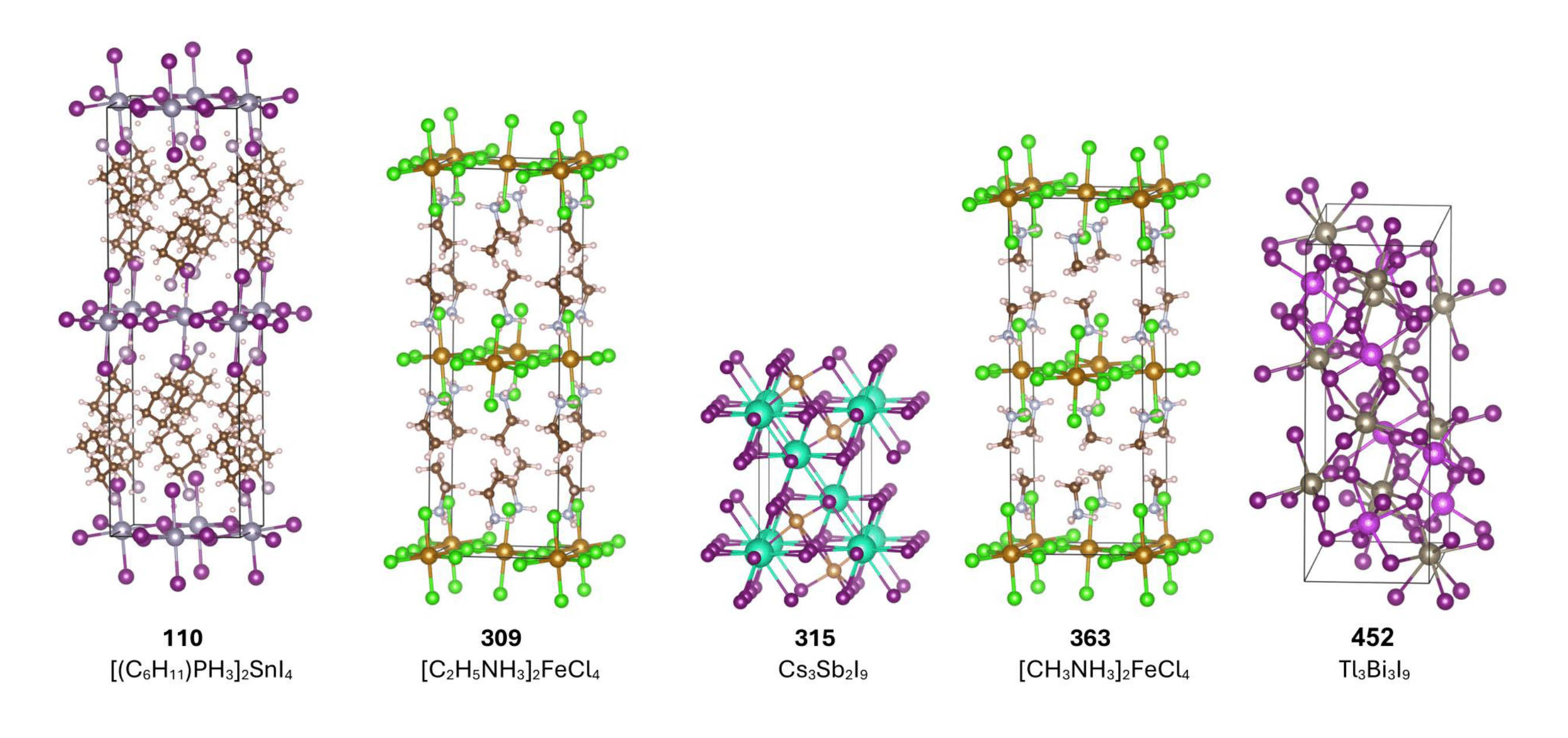}}
\caption{Five new 2D perovskite crystal structures, from the NMSE database explored through density functional theory (DFT) calculations. Visualization of the crystal structure is produced by the VESTA software~\cite{momma2011vesta}.}
\label{fig:The_5_materials}
\end{figure}

\subsection*{Main Phases of 2D Perovskites}
Two-dimensional (2D) layered halide perovskites (LHPs) are derived from the three-dimensional (3D) perovskite structure $\rm{A}\rm{B}\rm{X}_3$ where $\rm{A}$ is an organic cation, $\rm{B}$ is a divalent atom and $\rm{X}$ is a halide ion. By introducing a bulky organic cation such as ${\rm C}_3{\rm H}_7{\rm NH}_3^+$ (${\rm PA}^+$), ${\rm C}_4{\rm H}_9{\rm NH}_3^+$ (${\rm BA}^+$), and ${\rm C}_8{\rm H}_9{\rm NH}_3^+$  (${\rm PEA}^+$) in between the layers of inorganic octahedral, the bulk 3D perovskite structure can be “sliced” into layers. Broadly, LHPs are categorized into three distinct classes. Especially, we briefly introduce the Ruddlesden-Popper (RP), Dion-Jacobson (DJ), and Alternating-Cation-Interlayer (ACI) 2D Perovskites~\cite{gong2022layered} in this section (cf. Fig. \ref{fig:Phases_of_2D_Perovskites}).

\paragraph{Ruddlesden-Popper (RP) 2D Perovskites} These perovskites with a chemical formula $\rm{A}'_2\rm{A}_{n-1}\rm{B}_{n}\rm{X}_{3n+1}$ are characterized by their interlayer space being occupied by bulky or long-chain monovalent organic cation-pairs (e.g., ${\rm A}’= {\rm CH}_3({\rm CH}_2)_3{\rm NH}_3^+$, ${\rm C}_6{\rm H}_5{\rm CH}_2{\rm NH}_3^+$, and ${\rm 
 A} = {\rm CH}_3{\rm NH}_3^+$) leading to the half octahedral displacement between the adjacent 2D perovskite slabs. The interlayer distances in RP LHPs are generally the largest due to the bilayer structure of monovalent cation pairs. However, this also results in the presence of van der Waals gaps in the crystal structure, which are known to be relatively weak interaction forces compared with hydrogen bonding interactions. 

\paragraph{Dion-Jacobson (DJ) 2D Perovskites} These perovskites with the formula $\rm{A}'\rm{A}_{n-1}\rm{B}_{n}\rm{X}_{3n+1}$ feature divalent cations (e.g., $\rm{A}’= \rm{C}_5\rm{H}_{11}\rm{NCH}_2\rm{NH}_3^{2+}$, and $\rm{A} = \rm{CH}_3\rm{NH}_3^+$) that are vertically bonded with 2D perovskite slabs, resulting in well-aligned layered structures with no displacement between octahedral slabs. The interlayer distance in DJ LHPs can be much shorter than those in RP LHPs primarily because of the monolayer structures of their divalent organic cations. Moreover, DJ LHPs eliminate the van der Waals gap that exists in the RP LHPs, and the interlayer cations interact strongly with adjacent 2D perovskite slabs via hydrogen bonds, potentially increasing the structural stability of 2D LHPs. 

\paragraph{Alternating-Cation-Interlayer (ACI) 2D Perovskites} ACI 2D perovskites with the formula $\rm{A}'\rm{A}_{n}\rm{B}_{n}\rm{X}_{3n+1}$ (e.g., $\rm{A}’=\rm{C}(\rm{NH}_2)_3$, and $\rm{A}=\rm{CH}_3\rm{NH}_3$) adopt an alternating cation arrangement in the interlayer space, inducing an octahedral rotation in the final layered crystal structure. The interlayer distance in ACI LHPs can be much shorter than those in RP LHPs due to the tightly packed interlayers consisting of mainly modest or small cations. However, part of the van der Waals gap is still retained due to the cation $\rm {A}$-$\rm{A}’$ alternating arrangement. Despite this, the elimination or reduction of van der Waals gaps may increase the structural stability of 2D LHPs.

\subsection*{Density Functional Calculation}

Density Functional Theory (DFT) calculations were performed on selected 2D perovskite materials using the Vienna Ab initio Simulation Package (VASP). The Generalized Gradient Approximation (GGA) functional was utilized for quantum mechanical interactions, while the Projector Augmented Wave (PAW) method was employed for core and valence electron interactions. A Gaussian smearing width of $0.05$ eV was applied for Brillouin-zone integration, ensuring a self-consistent field energy criterion of $10^{-5}$ eV$/$atom. Volume and ionic relaxations were conducted using a conjugate-gradient optimization algorithm, with the force on each atom maintained below $0.01$ eV$/\AA$. The plane-wave energy cut-off was consistently set at 500 eV. A $k$-mesh, tailored to each perovskite lattice structure, was used for structural determinations and electronic band gap calculations. The structures of the five materials from the NMSE database, as explored through DFT calculations in this paper, are illustrated in Figure \ref{fig:The_5_materials}.

\begin{figure}
\centering
\centerline{\includegraphics[width=1\textwidth]{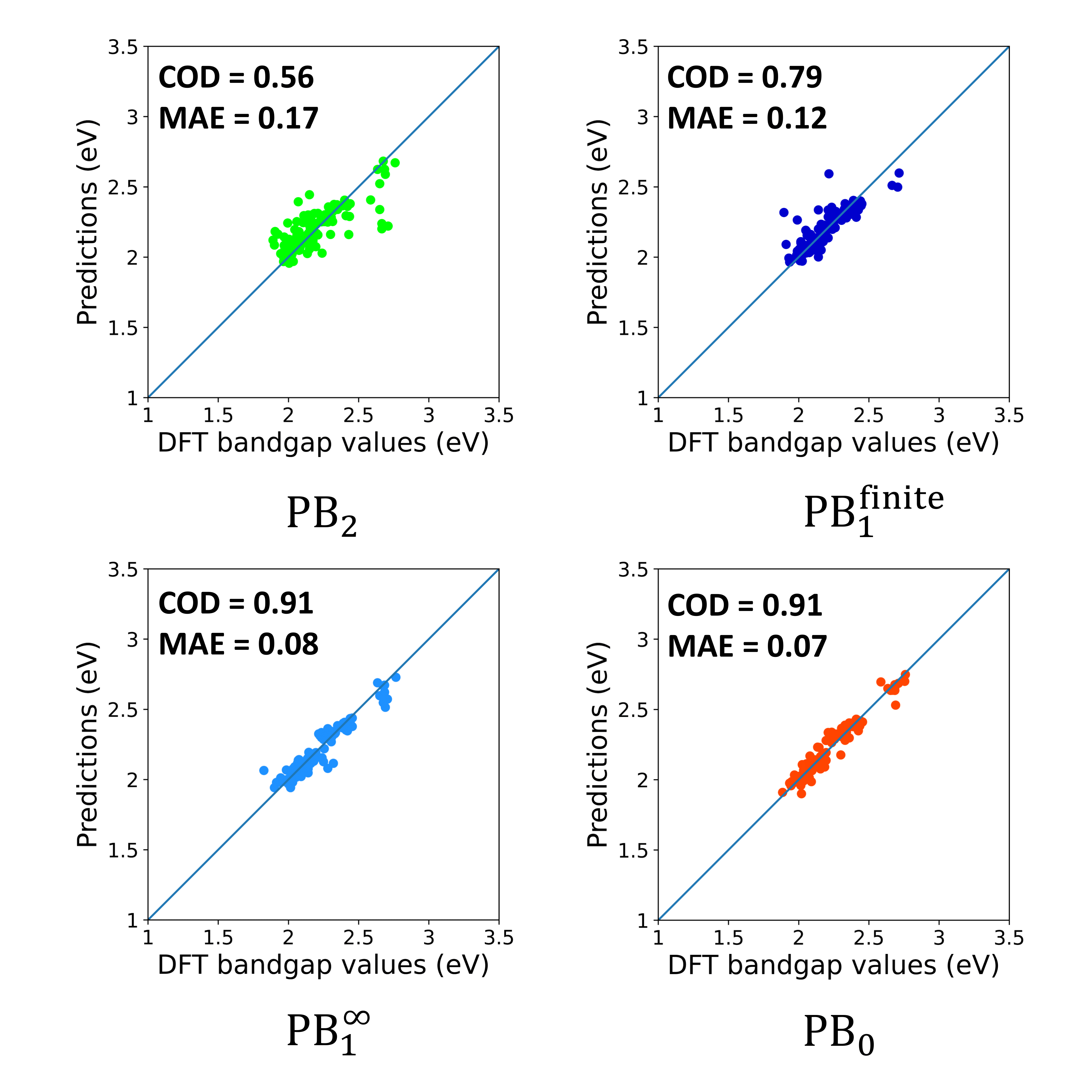}}
\caption{Visualized comparison of the performance of Betti curve features from ${\rm PB}_0$, ${\rm PB}_1^{\rm finite}$, ${\rm PB}_1^{\infty}$, and ${\rm PB}_2$ of the QC filtration $\overline{K_\bullet}$ in predicting the DFT-based bandgaps of 2D perovskites. Specifically, in the context of the 5-fold cross-validation encompassing $624$ materials, $125 \approx 624 \cdot 20\%$ test data points for each model are presented in images {\bf A}, {\bf B}, {\bf C}, and {\bf D}. The scatter plots in four diagrams depict the comparison between predictions and true values, with COD and MAE (eV) values rounded to the second decimal place based on the information in Table \ref{Table: Betti features comparison}. }
\label{fig:Betti_feature_comparison}
\end{figure}

\section{Further Experimental Results}

\subsection*{Feature Analysis on QCDs}

In order to gain deeper insights into the contributions of PB information in different dimensions, namely ${\rm PB}_0$, ${\rm PB}_1^{\rm finite}$, ${\rm PB}_1^{\infty}$, and ${\rm PB}_2$ of the QC filtration $\overline{K_\bullet}$, we conduct similar experiments in predicting DFT-bandgaps using their descriptors separately, employing GBT models. Table \ref{Table: Betti features comparison} depicts the result with COD, PCC, MAE (eV), and RMSE (eV) metrics, and Fig. \ref{fig:Betti_feature_comparison} utilizes the scatter diagrams to compare the predictions and true values. 
In order to simplify feature dimensions, the outcomes for ${\rm PB}_0$, ${\rm PB}_1^{\rm finite}$, ${\rm PB}_1^{\infty}$, and ${\rm PB}_2$ in Table \ref{Table: Betti features comparison} are derived from a focus on the Betti curve descriptors.

Based on the results presented in Table \ref{Table: Betti features comparison}, it is evident that the information from ${\rm PB}_1^{\infty}(\overline{K_\bullet})$ and ${\rm PB}_0(\overline{K_\bullet})$ in the QC filtration $\overline{K_\bullet}$ exhibits superior performance compared to ${\rm PB}_1^{\rm finite}(\overline{K_\bullet})$ and ${\rm PB}_2(\overline{K_\bullet})$. Particularly, considering~\eqref{Eq. PB1 decomposition}, it is noteworthy that ${\rm PB}_1^{{\rm finite}}(\overline{K_\bullet}) = {\rm PB}_1(K_\bullet)$ and ${\rm PB}_2(\overline{K_\bullet}) = {\rm PB}_2(K_\bullet)$ represent barcodes that essentially record the PH information within the original PH information. On the contrary, persistence intervals in ${\rm PB}_1^{\infty}(\overline{K_\bullet})$ are generated purely based on the periodicity of the crystal structure, while ${\rm PB}_0(\overline{K_\bullet})$ is constructed through the reduction of repeated intervals from the persistence barcode of the extended cell structure, both of which are influenced by the periodic and quotient structures. This observation underscores the potential of the proposed QC-based descriptors in effectively integrating both periodic and geometric information, thereby offering a comprehensive summarization of the material with enhanced insights into periodicity and geometry.

\begin{table}
\centering
\caption{Comparison of the performance of Betti curve features from ${\rm PB}_0$, ${\rm PB}_1^{\rm finite}$, ${\rm PB}_1^{\infty}$, and ${\rm PB}_2$ of the QC filtration $\overline{K_\bullet}$ in predicting the DFT-based bandgaps of 2D perovskites, based on the NMSE database~\cite{marchenko2020database}. Evaluation scores include the coefficient of determination (COD), Pearson correlation coefficient (PCC), mean absolute error (MAE), and root-mean-square error (RMSE). These metrics represent average scores obtained through 5-fold cross-validation conducted 5 times, with a total of 624 materials featuring DFT-based bandgaps in each cross-validation fold. The downstream model employed is the GBT model. The ``All QCDs'' result corresponds to the outcome of ``QC-GBT'' based on 624 materials as presented in Table 1 of the main document. This approach involves the holistic integration of all QCDs together, yielding a more promising result compared to considering individual QCDs in isolation.}
\begin{tabular}{|ccccc|}
 \hline
 Feature type & COD & PCC & MAE (eV) & RMSE (eV) \\ % [0.5ex]
 \hline\hline
 ${\rm PB}_2$ & 0.5551 & 0.7465 & 0.1741 & 0.2582 \\
 ${\rm PB}_1^{\rm finite}$ & 0.7934 & 0.8964 & 0.1206 & 0.1760  \\
 ${\rm PB}_1^{\infty}$ & 0.9061 & 0.9533 & 0.0772 & 0.1167 \\
 ${\rm PB}_0$ & 0.9086 & 0.9543 & 0.0744 & 0.1154 \\
 All QCDs & 0.9207 & 0.9613 & 0.0722 & 0.1086 \\
 \hline
\end{tabular}
\label{Table: Betti features comparison}
\end{table}

\begin{table}
\centering
\caption{Influence of Exp/DFT-based bandgap labels on QC-GBT models. The notation $X \Vert Y$ denotes the training set $X$ and testing set $Y$. There are 716 materials with DFT-based bandgaps and 235 materials with Exp-based bandgaps involved in these experiments. Evaluation scores are the coefficient of determination (COD), Pearson correlation coefficient (PCC), mean absolute error (MAE), and root-mean-square error (RMSE). The first row (Exp-based dataset $\Vert$ Exp-based dataset) presents the average scores based on the 5-fold cross-validation 5 times. On the other hand, for the second experiment (DFT-based dataset $\Vert$ Exp-based dataset), we train on the DFT-based dataset, test on the Exp-based dataset 5 times, and output the average scores.}
\begin{tabular}{|ccccc|}
 \hline
 Experiment (Traning $\Vert$ Testing) & COD & PCC & MAE (eV) & RMSE (eV) \\ % [0.5ex]
 \hline\hline
% DFT-based dataset $\Vert$ DFT-based dataset & 0.9214 & 0.9610 & 0.0731 & 0.1082  \\
 Exp-based dataset $\Vert$ Exp-based dataset & 0.8418 & 0.9257 & 0.1462 & 0.2176 \\
 DFT-based dataset $\Vert$ Exp-based dataset & 0.7268 & 0.8860 & 0.1889 & 0.2917  \\
 \hline
\end{tabular}
\label{Table: DFT vs Exp}
\end{table}

\subsection*{Experimental bandgap prediction}
%within the NMSE database (\url{http://www.pdb.nmse-lab.ru/})
To assess the predictive ability of DF features for Exp-based bandgaps in 2D perovskites, we conducted a 5-fold cross-validation on the Exp-based subset. We focus on the materials collected from 2021 to 2022 and use 716 materials with DFT-based bandgaps and 235 materials with Exp-based bandgaps as the experiment data. The results are presented in Table \ref{Table: DFT vs Exp}, demonstrating a notable contrast compared to the DFT-based subset.

One possible explanation for this difference is the relatively smaller size of the Exp-based subset, which comprises less than one-third of the DFT-based subset, thereby exerting a strong influence on the ML performance. Additionally, we trained the model using the DFT-based subset and evaluated its performance on the Exp-based subset. As expected, the results differed from the 5-fold cross-validation on the Exp-based subset. This discrepancy is reasonable, considering the larger size of the test set, which is expanded by a factor of 5, and the absence of training samples with Exp-based bandgap.

\end{document}